\RequirePackage{fix-cm}
\documentclass[twocolumn]{svjour3}          % twocolumn
\smartqed  % flush right qed marks, e.g. at end of proof
\usepackage{graphicx}
\usepackage{graphicx}
\usepackage{amsmath} 

\usepackage{xcolor}
\usepackage{algorithm}
\usepackage{algpseudocode}
\usepackage{array}
\usepackage{subcaption}

\usepackage{amsthm}

\usepackage{tabularx}
\usepackage{hyperref}
\usepackage{xurl}
\usepackage{enumitem}
\setlist[itemize]{labelsep=0.5em, leftmargin=1em}
\usepackage{threeparttable}
\begin{document}

\title{Chasing the Shadows: TTPs in Action to Attribute Advanced Persistent Threats%\thanks{Grants or other notes
%about the article that should go on the front page should be
%placed here. General acknowledgments should be placed at the end of the article.}
}
% \subtitle{Do you have a subtitle?\\ If so, write it here}

%\titlerunning{Short form of title}        % if too long for running head

\author{Nanda Rani         \and
        Bikash Saha \and
        Vikas Maurya \and
        Sandeep Kumar Shukla %etc.
}

%\authorrunning{Short form of author list} % if too long for running head

\institute{N. Rani \at
              Department of Computer Science \& Engineering, Indian Institute of Technology, Kanpur \\
              % Tel.: +123-45-678910\\
              % Fax: +123-45-678910\\
              \email{nandarani@cse.iitk.ac.in}           %  \\
%             \emph{Present address:} of F. Author  %  if needed
           \and
           B. Saha \at
              Department of Computer Science \& Engineering, Indian Institute of Technology, Kanpur \\
              % Tel.: +123-45-678910\\
              % Fax: +123-45-678910\\
              \email{bikash@cse.iitk.ac.in}
              \and
           V. Maurya \at
              Department of Computer Science \& Engineering, Indian Institute of Technology, Kanpur \\
              % Tel.: +123-45-678910\\
              % Fax: +123-45-678910\\
              \email{vikasmr@cse.iitk.ac.in}
              \and
           S. K. Shukla \at
              Department of Computer Science \& Engineering, Indian Institute of Technology, Kanpur \\
              % Tel.: +123-45-678910\\
              % Fax: +123-45-678910\\
              \email{sandeeps@cse.iitk.ac.in}
}

\date{Received: date / Accepted: date}
% The correct dates will be entered by the editor

\maketitle

\begin{abstract}
The current state of Advanced Persistent Threats (APT) attribution primarily relies on time-consuming manual processes. These include mapping incident artifacts onto threat attribution frameworks and employing expert reasoning to uncover the most likely responsible APT groups. This research aims to assist the threat analyst in the attribution process by presenting an attribution method named CAPTAIN (\textbf{C}omprehensive \textbf{A}dvanced \textbf{P}ersistent \textbf{T}hreat \textbf{A}ttr\textbf{I}butio\textbf{N}). This novel APT attribution approach leverages the Tactics, Techniques, and Procedures (TTPs) employed by various APT groups in past attacks. 
CAPTAIN follows two significant development steps: baseline establishment and similarity measure for attack pattern matching.
This method starts by maintaining a TTP database of APTs seen in past attacks as baseline behaviour of threat groups. The attribution process leverages the contextual information added by TTP sequences, which reflects the sequence of behaviours threat actors demonstrated during the attack on different kill-chain stages. 
Then, it compares the provided TTPs with established baseline to identify the most closely matching threat group. CAPTAIN introduces a novel similarity measure for APT group attack-pattern matching that calculates the similarity between TTP sequences. The proposed approach outperforms traditional similarity measures like Cosine, Euclidean, and Longest Common Subsequence (LCS) in performing attribution.
Overall, CAPTAIN performs attribution with the precision of $61.36\%$ (top-1) and $69.98\%$ (top-2), surpassing the existing state-of-the-art attribution methods.
% Insert your abstract here. Include keywords, PACS and mathematical
% subject classification numbers as needed.
% \keywords{First keyword \and Second keyword \and More}
% \PACS{PACS code1 \and PACS code2 \and more}
% \subclass{MSC code1 \and MSC code2 \and more}
\end{abstract}

\section{Introduction}
\label{sec:introduction}
Advanced Persistent Threats (APTs) pose a formidable challenge by employing sophisticated techniques, covert operations, and prolonged attacks on government organizations and critical infrastructures. These attacks are meticulously planned and performed by a group of attackers (usually sponsored by nation-states) working together for a common goal, usually called APT groups.
Identifying the responsible APT group for an attack is known as \emph{APT attribution}~\cite{Mei2022A,Steffens2020Attribution,sachidananda2023apter}.
Attribution holds significant value as it enables organizations to understand the motives and capabilities of APT groups. Moreover, it aids in effective defensive strategies and incident response investigations and promotes international collaboration to counter global cyber threats~\cite{Mei2022A,Steffens2020Attribution}.
This research attempts to empower defenders with an edge over adversaries by developing an automated system for uncovering the responsible entities behind sophisticated attacks. 

The attribution process involves analysis of various attack artifacts~\cite{alshamrani2019survey,Steffens2020Attribution}, such as malware samples, network traffic, system logs, and indicators of compromises (IOCs), and sometimes non-technical artifacts like geo-political consideration, telemetry to find the link between sophisticated attacks and associated threat actors.
Generally, the incident response team uses their expertise and experience to map the demonstrated modus operandi with the known threat group's behaviours~\cite{egloff2020public}.
In this mapping process, they leverage the traditional frameworks such as the diamond model~\cite{Ren2022CSKG4APT,Caltagirone2013Diamond}, which require extensive manual effort and expert reasoning~\cite{Mei2022A}. In addittion, the human-in-the-loop and manual analysis in the traditional threat attribution process are burdensome and susceptible to subjectivity. Therefore, there is a need for an automated method to reduce the efforts of the analysts and assist them. 

% According to the Pyramid of Pain (PoP)~\cite{Bianco2014Enterprise}, lower-level artifacts are easy for threat actors to alter(Fig.~\ref{fig:pop}), making methods based on these artifacts less robust . Higher-level artifacts, like Tactics, Techniques, and Procedures (TTPs), are the hardest for attackers to fabricate (Fig.~\ref{fig:pop}), providing robust ways to relate malicious activity compared to other IOCs. TTPs represent the threat actors' modus operandi, including attack vectors, tools, techniques, and behaviors~\cite{Daszczyszak2019TTP}. This concept inspires considering TTPs for threat attribution and offers confidence in understanding unique attacking patterns. TTP-based attribution~\cite{Bianco2014Enterprise, Noor2019A} motivates presenting an attribution method leveraging TTPs from past attacks. By examining TTPs, analysts can identify patterns, similarities to known threat groups, and unique signatures, enhancing the attribution process by narrowing down potential threat actors based on their behavior~\cite{Daszczyszak2019TTP}.

According to Pyramid of Pain~\cite{Bianco2014Enterprise}, lower-level artifacts are easy to alter for threat actors (shown in Fig.~\ref{fig:pop}), and therefore, methods based on these artifacts are not much robust. The higher-level artifact, i.e., Tactics, Techniques, and Procedures (TTPs), is the most challenging aspect for attackers to fabricate (shown in Fig.~\ref{fig:pop}) and, therefore, facilitates robust ways to relate malicious activity compared to other IOCs. TTPs represent the modus operandi of the threat actors, including their preferred attack vectors, tools, techniques, and operational behaviors~\cite{Daszczyszak2019TTP}. The Pyramid of Pain~\cite{Bianco2014Enterprise} inspires considering TTPs for threat attribution and gives confidence and an initial base for understanding the unique attacking patterns of threat groups. By examining TTPs, analysts can identify patterns, similarities to known threat groups, and signatures that differentiate the attacker's modus operandi~\cite{Daszczyszak2019TTP}.

\begin{figure}[ht]
    \centering
    \includegraphics[width=\columnwidth]{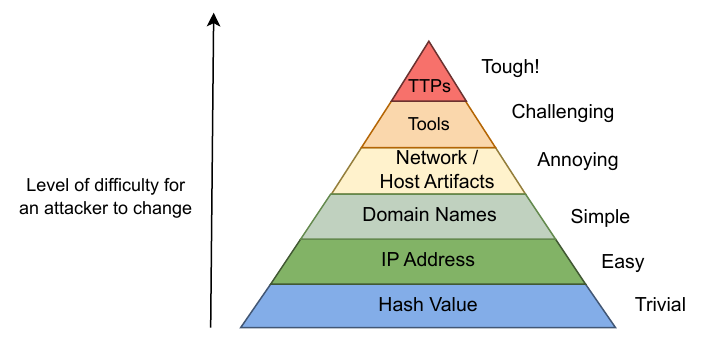}
    \caption{Pyramid of Pain~\cite{Bianco2014Enterprise}}
    \label{fig:pop}
\end{figure}

% The incident response team of security firms analyses a range of various artifacts to understand the modus operandi used by threat actors and structure them in terms of TTPs.
% Security firms privately maintain the database for seen artifacts of past attacks, including TTP. 

% As the attribution is performed campaign-wise, there is a dearth of substantial numbers of attack campaigns associated with specific APT groups.
Due to the stealthiness and longevity of the APT campaigns, there is a dearth of substantial numbers of attack campaigns associated with specific APT groups.
By referring to traditional and standard knowledgebases like the MITRE's APT group list~\footnote{Groups \url{https://attack.mitre.org/groups/}}, we can observe that the number of attack campaigns performed by individual APT groups is limited. 
As the attribution is performed campaign-wise, it creates a barrier for training-based methods that require a large number of samples per threat group to train the model. 
Consequently, there is a pressing need for a data-efficient method that can be developed without training on a huge sample set. Therefore, we present a method named CAPTAIN (\textbf{C}omprehensive \textbf{A}dvanced \textbf{P}ersistent \textbf{T}hreat \textbf{A}ttr\textbf{I}butio\textbf{N}), which is a TTP-based, data-efficient, and attack-pattern matching method to identify the associated threat group of given attack campaign.

% It shows limited data sample problems for automated attribution. 
% collecting numerous attack campaigns for any APT groups is impractical. Therefore, 
% It creates a barrier for traditional attribution methods as they require a large number of samples per threat group for training. 
% Therefore, the current attribution process heavily relies on traditional frameworks such as the diamond model ~\cite{Caltagirone2013Diamond,Ren2022CSKG4APT}, which require extensive manual effort and expert reasoning~\cite{Mei2022A}. The human-in-the-loop and manual analysis in the traditional threat attribution are burdensome and susceptible to subjectivity.
% Consequently, there is a pressing need for a data-efficient method 
% to enhance the attribution process and identify attack patterns 
% Consequently, there is a pressing need for a data-efficient method , which can be developed without training on a huge sample set. Therefore, we present a method named CAPTAIN (\textbf{C}omprehensive \textbf{A}dvanced \textbf{P}ersistent \textbf{T}hreat \textbf{A}ttr\textbf{I}butio\textbf{N}), which is a TTP-based, data-efficient, and attack-pattern matching method to identify the associated threat group of given attack campaign.
% for threat attribution by analyzing TTPs observed in the attack and employing similarity pattern matching.

% CAPTAIN is a TTP-based, data-efficient, and attack-pattern matching method to identify the associated threat group of given attack campaign. 
% Based on threat actor's motivation and expertise, 
The threat posed by APT groups involves executing a series of sophisticated attack vectors over an extended period, following the life cycle of APTs containing several kill-chain stages~\cite{che2024systematic,alshamrani2019survey,sharma2023advanced}. The order in which the attack vectors are used, or the sequence of TTPs observed in the kill-chain phases, provides additional contextual information and illustrates the threat actor's behavior throughout the attack. The sequential information also reflects their choice of attack vectors, tools, and techniques for a specific kill-chain phase. The contribution of sequences in terms of their correlation is experimented with and discussed in Section~\ref{subsec:SeqCorr}.
The TTP-based attribution method in the literature employs the mere presence of the TTPs and misses the contextual characteristics of the set of TTPs observed~\cite{Noor2019A}. CAPTAIN leverages the contextual information provided by the sequence of behaviours seen, i.e., the sequence of TTPs employed during the attack. It examines the sequence of techniques that APT groups exhibit during different attack phases and how frequently they repeat patterns of TTPs within a single attack campaign.

Given a set of TTPs observed in the attack campaign, CAPTAIN assists the incident response team by sequencing them based on kill-chain for performing attribution and provides the most probable linked threat group based on known and observed attack patterns in terms of TTPs. Hence, the input of the CAPTAIN is a set of TTPs, and the output is the attributed threat group.
Along with the limited group-wise attack campaigns, there is a notable gap in the availability of a publicly structured TTP-based attribution dataset which restricts the development of attribution methodologies. The public data is available primarily in the form of threat intelligence reports published by security firms, which can be leveraged to prepare the TTP database. These threat reports contain end-to-end attack flow information in natural language, explaining the attacker's modus operandi. Therefore, we leverage publicly available threat reports to prepare the TTP database in this experiment to perform attribution.
There are several methods that can translate the explained TTPs in the threat report sentences into their corresponding MITRE ATT\&CK TTPs~\cite{Husari2017Ttpdrill,alam2023looking,TTPXHunter2024Nanda,Rani2023TTPHunter} and aid in preparing the TTP-database for attribution.

The proposed method starts with maintaining the TTP database of known APT groups as a baseline.
The database is prepared by leveraging the past attack campaign's threat report and employing our tool named TTPXHunter~\cite{TTPXHunter2024Nanda}, an extended version of our previous tool named TTPHunter~\cite{Rani2023TTPHunter}, for TTP extraction from threat reports.
CAPTAIN creates the TTP sequence from the extracted set of TTPs based on attack kill-chain phases. 
It combines MITRE ATT\&CK framework~\cite{MITREATTACK} and Unified-Kill-Chain (UKC) model~\cite{Pols2017The} to arrange the observed TTPs in the sequence of attack flow for each attack sample present in the dataset.
After the TTP extraction and sequencing, we perform pattern matching of attack patterns with an established baseline. We propose a novel similarity measure for pattern matching that can capture subsequence similarity based on the length of matched subsequence patterns and their frequency within the sequence.
This measure allows us to match the TTP sequence and attribute APT groups for a given attack campaign. 
% For this experiment, we leverage our tool named TTPXHunter~\cite{TTPXHunter2024Nanda}, an extension of TTPHunter~\cite{Rani2023TTPHunter}, to extract the past attack campaign's TTPs of threat groups. 
% CAPTAIN creates  the TTP sequence from extracted set of TTPs based on attack kill-chain phases. 
% It combines MITRE ATT\&CK framework~\cite{MITREATTACK} and Unified-Kill-Chain (UKC) model~\cite{Pols2017The} to arrange the observed TTPs in the sequence of attack flow for each attack sample present in the dataset.
% After the TTP extraction and sequencing, we perform pattern matching of attack patterns with established baseline. For pattern matching, we propose a novel similarity measure that can capture subsequence similarity based on the length of matched subsequence patterns and their frequency within the sequence.
% To observe the behavioural patterns of attackers in the TTP sequences database, we propose a novel similarity measure that can capture subsequence similarity based on the length of matched subsequence patterns and their frequency within sthe sequence. 
% This measure allows us to match the TTP sequence and attribute APT groups for a given attack campaign. 

During the testing phase, CAPTAIN creates the TTP sequence for a set of TTPs as a given test sample. Then, it uses the proposed similarity measure to compute the similarity between the test sample and the TTP sequences present in the established baseline database. The threat group whose TTP sequences are more similar to the given set of TTPs is the attributed threat group. 
We evaluate the performance of CAPTAIN by preparing a TTP database from past attack campaign details maintained in the security firm's threat reports. We compare the proposed similarity measures with traditional similarity measures and available attribution methods present  in the literature. The results demonstrate that the proposed similarity measure outperforms the traditional similarity measures and available attribution methods. 
% CAPTAIN assists the incident response team in attributing the threat actor once they structure the modus operandi in terms of TTPs observed during the attack. When provided with a set of TTPs as input, it outputs the threat group whose sequence of TTPs aligns best with those in the database.
% Given the input as set of TTPs it provides the output as a known threat group whose TTP sequences are more aligned with the maintained database.

% The performance of CAPTAIN over traditional similarity measures and state-of-the-art methods contributes towards addressing RQ1 and demonstrates the effectiveness of TTPs in APT attribution. The proposed similarity measure contributes towards addressing the RQ2 and facilitates capturing other characteristics, such as length and frequency of attack patterns, rather than the mere presence of TTP.

The key contributions of this research include the following: 
\begin{itemize}
    \item We propose a novel attribution method named CAPTAIN~\footnote{We plan to make it available publicly for further research progress after the research paper is accepted.}, which utilizes the flow of the attacker's behaviour in the form of TTP sequences.
    The method prepares TTP sequences by leveraging the kill-chain framework and employs attack patterns matching using the past attack campaign's TTP sequences used by APT groups.
    \item We prepare a structured TTP-based attribution dataset~\footnote{\url{https://github.com/nanda-rani/CAPTAIN} - This is the only Glimpse of the dataset. The remaining part, including the code, will be released after the research paper is accepted.} using our tool TTPXHunter~\cite{TTPXHunter2024Nanda}, an extended version of our previous tool named TTPHunter~\cite{Rani2023TTPHunter}, to facilitate the attribution of APT attacks. The dataset encompasses the $580$ TTP sequences observed in past attack campaigns belonging to $11$ APT groups. 
    \item We propose a novel similarity measure for the attribution process. It considers matching TTP sub-sequences and their frequency to establish links with the APT group's behaviour stored in the database. This measure is inspired by Longest Common Subsequence (LCS)~\cite{Hirschberg1977Algorithms} and Gelstat Pattern Matching~\cite{Ratcliff1988Pattern} algorithms.
    \item We perform experiments to compare our proposed similarity measure with traditional similarity measures for performing attribution. These experiments demonstrate the comparison between attribution performance based on mere presence of TTPs and TTP sequences. We also implement an existing literature to compare performance of CAPTAIN with the state-of-the-art attribution method.
    % and existing attribution methods. Overall, the CAPTAIN performs better than each of the implemented models.
\end{itemize}

The remainder of the paper is as follows: Section~\ref{sec:RelatedWork} presents a thorough review of related work. Section~\ref{sec:background} provides the essential background information needed to understand the proposed methodology. We comprehensively explain our proposed method CAPTAIN, including the baseline establishment and similarity measure for attack pattern matching in Section~\ref{sec:CAPTAIN}. We delve into the intuition behind our proposed similarity measure and provide necessary proofs in Section~\ref{subsec:IntuttionProof}. Further, in Section~\ref{sec:ExperimentResult}, we elaborate on the experiments conducted, their results, evaluation, and noteworthy observations. We discuss the challenges faced and limitations in Section~\ref{sec:ChallengeLimitation}, while Section~\ref{sec:conclusion} concludes our research and discusses potential avenues for future work.

\section{Related Work}
\label{sec:RelatedWork}

% This research focuses to uncover the responsible threat group and this section presents an extensive literature  and dataset review. We also perform comprehensive comparison between literature and presented methods.

This section presents an extensive related work and available database in the literature for uncovering the responsible threat group. 
We also compare the attribution method present in the literature and proposed method.

\subsection{Attribution Data in Literature}
We explore the public domain to find the APT group's modus operandi dataset, which can be used to perform attribution. There is an extensive knowledgebase of the APT group's modus operandi presented by MITRE~\cite{MITREgroup}. This knowledgebase contains all sets of modus operandi that have been seen for a specific threat group to date. However, attribution is a process of identifying the most likely responsible threat group of an attack campaign~\cite{Steffens2020Attribution}. So, for any algorithm to understand the threat actor's way of attack, it may need the data sample, i.e., modus operandi, attack-wise rather than all together.
% MITRE also maintain campaigns lists~\cite{MITREcampaign} in their knowledgebase, which may be helpful for an algorithm to understand the threat actor's modus operandi attack-wise, but as of data of writing their knowledgebase only have $30$ campaigns details listed. Given in increase in number of campaigns in future, their knowledgebase can be a standard and trusted database which can aids in attribution methodology research based modus operandi. 
Our exploration also found prominent security firm's threat reports, which contain information regarding end-to-end attack analysis, how the attack started, how it progressed, and how it was detected. The past attack threat reports can be a great resource for collecting the modus operandi of threat actors seen in past attacks due to their comprehensiveness. Moreover, these reports are present in unstructured natural language form, which needs to be processed to transform into structured data.
% explain the modus operandi employed by threat actors in the APT incidents. These reports are published by several prominent security firms and 
% contain information regarding end-to-end attack analysis, how the attack started, how it progressed, and how it was detected. The past attack threat reports can be a great resource for collecting the modus operandi of threat actors seen in past attacks due to their comprehensiveness. Moreover, these reports are present in unstructured natural language form, which needs to be processed to transform into structured data.

% Further, these details are present in unstructured natural language form, which needs to be processed to transform into structured data.
We find that the lack of structured APT group's modus operandi dataset hinders threat group's behaviour analysis. We address this gap by leveraging past attack campaign threat reports to prepare structure dataset. To create such a dataset from threat reports, we need certain tool to interpret and translate threat actor's modus operandi from natural language into TTPs.
% Considering threat reports as a source to prepare the structured dataset, we need certain tools to understand the threat actor's modus operandi explained in natural language and translate it in the form of TTPs. 
There are various research present based on keyword-phrase matching, ontology, graphs, and language models in the literature to translate TTPs from natural langugae text~\cite{Husari2017Ttpdrill,alam2023looking,TTPXHunter2024Nanda,Rani2023TTPHunter}. 
In our previous work, we present language based tool named TTPHunter~\cite{Rani2023TTPHunter} to map relevant sentences to MITRE ATT\&CK TTPs and we further extend it by presenting state-of-the-art TTPXHunter~\cite{TTPXHunter2024Nanda} which is based on cyber domain-specific language model and outperforms other state-of-the art methods.
% We also propose a natural language model-based tool named TTPHunter~\cite{Rani2023TTPHunter} to understand the sentence written in natural language and map it to the MITRE ATT\&CK TTPs. This model is limited to only frequently used $50$ TTPs due to unavailability of sentence dataset. We enhanced capability of TTPHunter and presented a state-of-the-art model named TTPXHunter~\cite{TTPXHunter2024Nanda} to extend the TTP identification capability of the model with an impressive array to $192$ TTPs.
Therefore, we leverage our tool TTPXHunter~\cite{TTPXHunter2024Nanda} to extract MITRE ATT\&CK TTPs in the structured format from past attack threat reports and prepare the attribution dataset.

\subsection{Attribution Methodology}
\label{subsec:AttributionMethodLiterature}

In the current literature, attribution is performed using traditional frameworks or \sloppy
automated methods. The automated threat attribution methods majorly based on leveraging malware sample, past threat analysis reports, and modus operandi in terms of TTP.

% Manzano et. al~\cite{gonzalez2023technical} present the differentiation and characterization between regular malware and APT malware. Which can be considered as base for exploring the possibility of identifying threat group using APT malware.
Kida et al.~\cite{Kida2022Nation} propose malware-based attribution that feeds fuzzy hashes into machine learning-based classifiers as natural language input to attribute APT group. Wang et al.~\cite{Wang2021Explainable} identify groups based on malware's string and code features. They feed both types of features to the classifier to attribute the group and employ Local Interpretable Model-agnostic Explanations (LIME) to explain which feature aids in identifying the group. Rosenberg et al.~\cite{rosenberg2017deepapt} employ raw dynamic features extracted from cuckoo sandbox analysis reports, treating each unique word as a distinct feature, and use one-hot encoding to prepare their dataset. They apply deep neural networks to detect patterns within these features for classifying APT groups. Notably, their model recognizes the associated nation instead of the specific threat group, resulting in less precise in-group identification. The number of sophisticated malware is increasing rapidly and offers machine learning-based analysis. 
However, malware-based attribution often faces two problems: \textbf{(P1)} Malware can be modified, shared, or sold, leading to significant challenges in pinpointing the true source of a cyber attack~\cite{gray2024identifying}. This approach struggles to provide a consistent view of threat actors when they employ different malware across various campaigns. \textbf{(P2)} The reliance on malware signatures can miss coordinated, multi-faceted attacks that involve a range of attack vector and tools and don't rely on single malware instance.

% However, malware-based attribution often falls short due to the ease with which malware can be modified, shared, or sold, leading to significant challenges in pinpointing the true source of a cyber attack. This approach struggles to provide a consistent view of threat actors when they employ different malware across various campaigns. Furthermore, the reliance on malware signatures can miss coordinated, multi-faceted attacks that involve a range of tactics and tools, reducing the overall effectiveness of cyber threat analysis.

% However, malware-based attribution methods face these challenges in attribution: 1) Attribution method may get misleads if the APT groups purchase the malware from an illicit dark market, 2) There can be situation where the exact malware may not be get captured or used by adversary, such as file-less malware or conducts attacks in real-time through commands after achieving persistence, 3) APTs used series of malware in an campaign to completed their prolonged attack. An individual malware may not have all kill-chain stages involved and therefore analysing them individually may not capture the whole attack life cycle.

Perry et al.~\cite{perry2019no} introduce NO-DOUBT, a machine learning approach for attributing attacks using textual analysis of threat intelligence reports. They present, SMOBI (Smoothed Binary vector), an improved bag of words which provides a similar vector for semantically similar words. Along with identifying known threat actors, they also introduce identifying novel threat actors by assessing the prediction probabilities of the trained model.
Further, Naveen et al.~\cite{Naveen2020Deep} improved the feature representation method named SIMVER. Their method is similar to~\cite{perry2019no}, in which they replace the bag of words with word2vec and assign similar word2vec embeddings to similar words, where the similarity between words is calculated using cosine measure. 
Both~\cite{Naveen2020Deep,perry2019no} convert embeddings of natural language text in the report using their corresponding method, i.e., SMOBI or SIMVER, and employ a classifier to attribute the group for any given threat report. Both of these methods focus on utilizing words present in threat reports as features. Considering words present in the report suffers two  problems:  \textbf{(P3)} Words comparison represents more like performing threat report similarity rather than the similarity between threat group's characteristics. \textbf{(P4)} Given P3, identifying features contributing towards attribution classification is essential to understand their impact and build a reliable attribution model.
% employs threat report text as the main feature, which may not effectively identify threat actor behavior. 
% Liu et. al~\cite{liu2022two} explore methods for identifying APT using two statistical features of network traffic features—C2Load\_fluct (measures response packet load fluctuations) and Bad\_rate—to (tracks the rate of problematic packets) for the APT identification. They tested these features using datasets with DNS and TCP traffic, demonstrates the effectiveness in distinguishing between normal and malicious traffic of APTs.

Using machine learning and deep learning, Noor et al.~\cite{Noor2019A} present TTP-based threat attribution. They leverage a semantic-based text matching to extract TTPs from threat reports and convert the extracted TTPs in one-hot encoding to get the correlation matrix as a feature vector. Further, they implemented various models and find that DLNN (Deep Learning Neural Network) performed better than other implemented models. Their dataset comprises $324$ samples distributed over a $36$ APT group. The group with a maximum number of samples consists of $23$ samples, and the group with a minimum number of samples contains only $3$ samples. On average, they report that their dataset consists of $9$ samples for each threat group. The highly imbalanced and limited dataset may lead to bias problems in the trained DLNN model. Also, for TTP extraction, their method relies on semantic search rather than contextual information represented by sentences in the report.
\\
\\
A comparison of literature with our proposed method CAPTAIN is shown in Table~\ref{tab:APT_LR}. The available literature methods for APT attribution consist of several issues, such as non-reliable malware-based attributions, limited TTP datasets, and limited attribution methods. Our proposed method, CAPTAIN, addresses discussed limitations by leveraging the attacker's operational behaviour pattern matching. CAPTAIN leverages TTPs as a feature, encompassing all attack vectors, tools, and techniques threat actors use. As the attribution needs to be done attack campaign-wise~\cite{Steffens2020Attribution}, collecting a dataset with numerous attack behaviour samples for any specific group is challenging. Therefore, we focus on pattern-matching methods rather than learning through training methods. CAPTAIN presents a TTP-based, data-efficient, and attack pattern-matching method for APT attribution.
% In addition, our method leverages TTPs as feature, which encompasses all attack vectors, tools and techniques used by threat actors.
% Our proposed method, CAPTAIN, deals with these limitations by leveraging the attacker's operational behaviour pattern matching.

\begin{table*}[!ht]
    \centering
    \caption{Attribution Methodology Literature Review Comparison}
    % FIXFIX \fix{Think some more basis to compare}}
    \label{tab:APT_LR}
    \begin{tabular}{p{1.5cm}p{1.2cm}p{1.5cm}p{1.3cm}p{10.5cm}}
        \hline
         \textbf{Methods} & \textbf{Artifact} & \textbf{Feature} & \textbf{Model} & \textbf{Comment}  \\
         \hline
         % \hline
         Kida et al.~\cite{Kida2022Nation} & Malware-based & Static Fuzzy hash with NATO Encoding & Random Forest Classifier & 
         \begin{itemize} [nosep, leftmargin=*]
         \vspace{-0.2cm}
            \item Leverage NATO encoded fuzzy hash of malware sample and pass it to the classifier for attribution
            \item The authors reported their method's performance struggles to surpass the state-of-the-art, aiming to understand how fuzzy hashing contributes to identifying nation-state actors.
            \item As the method is leveraging malware sample to attribute threat groups, it faces problems \textbf{P1} and \textbf{P2} as discussed in~\ref{subsec:AttributionMethodLiterature}
            \vspace{-0.4cm}
         \end{itemize}
         \\
         \hline
         Rosenberg et. al~\cite{rosenberg2017deepapt} & Malware-based & Unique words from Cuckoo Sandbox reports & Deep Neural Network Classifier & 
         \begin{itemize}[nosep, leftmargin=*] 
         \vspace{-0.2cm}
            \item Use words present in the Cuckoo sandbox reports, convert them into one-hot encoded vectors for attribution classification
            \item Attributes at the country level rather than the threat group level
            \item Word semantics of Cuckoo report reflect malware functionality rather than threat actor behavior
            \item As the method is leveraging malware sample to attribute threat groups, it faces problems \textbf{P1} and \textbf{P2} as discussed in~\ref{subsec:AttributionMethodLiterature}
            % \item Lacks trustworthiness due to the sale of malware on the dark web
            % \item Attribution level is country rather than threat group
            % \item Semantics of words present in cuckoo represent malware behaviour or functionality rather than threat actor's behaviour
            % \item A single malware may not cover the complete kill chain malicious activities
            \vspace{-0.4cm}
        \end{itemize}
        \\
        \hline
         Wang et al.~\cite{Wang2021Explainable} & Malware-based & Code and String Features & Deep Neural Network Classifier  & 
         \begin{itemize}[nosep, leftmargin=*] 
         \vspace{-0.2cm}
             \item Employ strings and intermediate code features as natural language and forward it to the classifier after getting feature vector using bag-of-words and paragraph vector
             \item Rely on static features for the decision which is challenging in the case of obfuscated samples. The obfuscation is common in APT types of attacks
             \item As the method is leveraging malware sample to attribute threat groups, it faces problems \textbf{P1} and \textbf{P2} as discussed in~\ref{subsec:AttributionMethodLiterature}
            % \item Lacks trustworthiness due to the sale of malware on the dark web 
            % \item A single malware may not cover the complete kill chain malicious activities
            \vspace{-0.4cm}
            \end{itemize}
            \\
        \hline
        Perry et. al~\cite{perry2019no} & Threat-report & SMOBI text embeddings & XGBoost Classifier  & 
        \begin{itemize}[nosep, leftmargin=*] 
        \vspace{-0.2cm}
            \item Their method uses an improved bag-of-words approach called SMOBI to convert threat report texts into feature vectors and pass them to a classifier for performing attribution
            \item They use report text as a feature, better for identifying similar reports than pinpointing threat actor behavior. SMOBI focuses on word frequency and similarity, ignoring semantics and context
            \item As the method is leveraging words present in the threat report to attribute threat groups, it faces problems \textbf{P3} and \textbf{P4} as discussed in~\ref{subsec:AttributionMethodLiterature}
            % \item Identifying features contributing to attribution classification is essential to understanding their impact and building a reliable attribution model.
            % \item Used feature is better at identifying similar reports than pinpointing threat actor behavior
            % \item SMOBI based on bag-of-the words which completely ignores the semantics and contextual meaning and focus on frequency of words
            % \item The words or feature contributing towards classification is needs to get uncovered to build reliable model
            \vspace{-0.4cm}
        \end{itemize}
        \\
         \hline
         Naveen et al.~\cite{Naveen2020Deep} & Threat-report & Threat Report Text Embedding, an extended version of SMOBI~\cite{perry2019no} & Deep Neural Network Classifier & 
         \begin{itemize}[nosep, leftmargin=*] 
         \vspace{-0.2cm}
             \item Builds on~\cite{perry2019no}, replacing the base bag-of-words with word2vec to enhance semantic meaning in feature vectors from threat report texts
             \item Utilize report texts as a feature, better for identifying similar reports than pinpointing threat actor behavior. SIMVER focuses on word semantics but ignores the contexts
             \item As the method is leveraging words present in the threat report to attribute threat groups, it faces problems \textbf{P3} and \textbf{P4} as discussed in~\ref{subsec:AttributionMethodLiterature}
            % \item Identifying features contributing to attribution classification is essential to understanding their impact and building a reliable attribution model.
            % \item Used feature is better at identifying similar reports than pinpointing threat actor behavior
            % \item The text embeddings is based on bag-of-the words which completely ignores the semantics and contextual meaning and focus on frequency of words
            % \item The words or feature contributing towards classification is needs to get uncovered to build reliable model
            \vspace{-0.4cm}
        \end{itemize}
        \\
         \hline
         Noor et al.~\cite{Noor2019A} & TTP-based & Correlation Matrix & Deep Neural Network Classifier & 
         \begin{itemize}[nosep, leftmargin=*] 
         \vspace{-0.2cm}
            \item Extract TTPs from threat reports using a semantic-based search method and pass the one-hot encoded TTPs to a classifier
            \item  Model is trained on the highly imbalanced and limited dataset, and their feature transformation method considers only the existence of TTPs and ignores contextual information such as when and how it has been used
            \item Their TTP extraction method is based on semantic  similarity rather than the contextual meaning of sentences
            % \item Consider only existence of TTPs and ignored contextual information such as when and how it has been used
            \vspace{-0.4cm}
        \end{itemize}
        \\
         \hline
         \textbf{CAPTAIN (Proposed)} & \textbf{TTP-based} & \textbf{TTP Sequence features} & \textbf{Novel Similarity Method} & 
         \begin{itemize}[nosep, leftmargin=*] 
         \vspace{-0.2cm}
            \item \textbf{A novel way to leverage contextual information of modus operandi by sequencing them based on kill-chain}
            \item \textbf{Based on attacker's modus operandi which can directly hint about their preferred choice and attack vectors}
            \item  \textbf{Employ attack-pattern similarity matching method to identify threat actor's behavioural patterns per group}
            \vspace{-0.4cm}
        \end{itemize}
        \\
         \hline
    \end{tabular}

\end{table*}

\section{Background}
\label{sec:background}

This section aims to establish a fundamental understanding of the key concepts and tools vital for cyber threat analysis and attribution, including models and frameworks that dissect and categorize cyber attack stages and strategies.
% The section also highlights the role of advanced natural language processing tools in analyzing extensive unstructured data. It covers specialized tools and methods designed for extracting threat intelligence in terms of TTPs used by attackers.

\subsection{Unified-Kill-Chain (UKC) model}
\label{subsec:ukc}

\begin{table*}
\caption{Notations and their description}\label{tab:parameters}
\centering
\begin{tabular}{ll}
\hline
 \textbf{Notation}  & \textbf{Description}\\
\hline 
% \hline
% $\mathbb{R}$ & set of Real numbers\\
$I^k$ & $k$ length integer vector\\
$\lambda$ & $k$ length integer vector ($\lambda \in I^k $) of common subsequences between two TTP sequences\\
$\mu$ & frequency vector ($\mu \in I^k $) of common subsequences corresponding to elements of $\lambda$\\
$\langle \mu,\lambda \rangle$ & dot product operation of two vectors $\mu$ and $\lambda$\\
% $a*b$ & Multiplication operation of corresponding elements of\\
% & two vectors $a$ and $b$\\
$seq_i$ & sequence of TTPs in $i^{th}$ threat report\\
$Sim(seq_i,seq_j)$ & similarity score between $seq_i \text{ and } seq_j$\\
$n$ & length of $seq_i$\\
$m$ & length of $seq_j$\\
% $w$ & Weight vector ($w=\mu * \lambda, w \in I^k $) of common \\ & subsequeces\\
$A$ & attack campaign as a test sample\\
$\mathcal{T}^{(g)}$ & all TTP sequence samples of $g^{th}$ group present in database \\
$T_i^{(g)}$ & $i^{th}$ TTP sequence sample of $g^{th}$ group\\
$N^{(g)}$ & number of TTP sequence samples in $g^{th}$ group\\
% \vspace{-0.25cm}
% &\\
$N$ & total number of TTP sequence samples present in the database\\
$\mathcal{S}(A,\mathcal{T}^{(g)})$ & attribution score of test sample $A$ for APT group $g$ \\
$\hat{g}$ & attributed group \\
$G$ & total number of APT groups present in the database \\
% $\mathcal{C}$ & Categorization of TTPs based on UKC \\
% $UKC$ & list of UKC attack phase $(ukc \in UKC)$ \\
% $R$ & set of threat reports \\
% $E(r_i)$ & TTP list extracted from report $r_i  (ttp \in E(r_i))$ \\
\hline
\end{tabular}
\end{table*}

The UKC model offers a comprehensive understanding of the tactics employed by APTs and ransomware groups~\cite{Pols2017The}. By amalgamating and extending existing models, this unified approach sheds light on the sequential order of attack phases from initiation to end. The UKC model comprises 18 attack phases and the specific sequence for each stage. 
Author Paul Pols highlights the limitations of the traditional Cyber Kill Chain (CKC)~\cite{Hutchins2011Intelligence} proposed by Lockheed Martin, which primarily focuses on parameters and malware, neglecting other attack vectors beyond organizational boundaries.
In contrast, the UKC model addresses these shortcomings by offering significant improvements and adopting a time-agnostic approach to tactics, as seen in MITRE's ATT\&CK framework. Furthermore, Pols incorporates pivoting within the UKC model to illustrate the pivotal role of choke points in attacks. Overall, the UKC framework~\cite{Pols2017The} facilitates a structured analysis and enables a deeper understanding of the attacker's tactical modus operandi. 
% \begin{figure}[h]
%   \centering
%   \includegraphics[width=\columnwidth]{fig/UKC.pdf}
%   \caption{Unified Kill Chain Model~\cite{Pols2017The}}
%   \label{fig:ukc}
% \end{figure}

% \vspace{-0.4cm}
\subsection{MITRE ATT\&CK Framework}
\label{subsec:MITREframework}
The MITRE ATT\&CK framework~\cite{MITREATTACK}, recognized globally, serves as an extensive knowledge base and framework for understanding cyber threats and adversary behaviors across different attack stages. It aims to bolster organizational capabilities in detecting, preventing, and responding to cyber threats by providing a standardized taxonomy for the tactics, techniques, and procedures (TTPs) used by attackers. The framework encapsulates Tactics, depicting adversary's objectives; Techniques, detailing the methods for achieving tactics; and Procedures, offering granular insights into how specific techniques are employed to fulfill tactical goals. With 14 tactics ranging from Reconnaissance to Impact and encompassing techniques like phishing and PowerShell exploitation, the ATT\&CK framework illustrates real-world attack scenarios. It facilitates enhanced defensive strategies, incident response, and threat intelligence for cybersecurity practitioners of both offensive and defensive.

\subsection{TTPXHunter}
\label{subsec:TTPXHunter}
% \fix{Add TTPXHunter part from methodology section}

TTPXHunter is an extension of our previous tool named TTPHunter~\cite{Rani2023TTPHunter} based on natural language-model for identifying TTPs explained in threat report sentence. TTPHunter is based on traditional BERT model~\cite{devlin2018bert} and limited to $50$ frequently used TTPs due to lack of sentence-TTP dataset. In TTPXHunter, we introduce domain-specific language model SecureBERT~\cite{Aghaei2023SecureBERT} and extend the capability of model to identify an impressive array of $193$ distinct TTPs present in MITRE ATT\&CK knowledgebase.
% It utilizes the concept of leveraging Natural Language Processing (NLP) techniques to extract TTPs from threat reports~\cite{TTPXHunter2024Nanda}. It is based on a domain-specific language model named SecureBERT~\cite{Aghaei2023SecureBERT} and 
TTPXHunter is a fine-tuned domain-specific language model to map the sentences to the relevant TTPs they represent. 
One noteworthy feature of TTPXHunter is its ability to discern the significance of sentences present in the threat report. Rather than indiscriminately mapping all sentences from a threat report, TTPXHunter selectively identifies and considers only those sentences that truly explain TTPs, disregarding irrelevant information. The TTPXHunter implements a filtering mechanism on the classification model's predicted probabilities to achieve this. The experiments done to evaluate the trained model discovered that sentences with a probability higher than $0.644$ are relevant, while those below the threshold were considered irrelevant. We adopt the same threshold to filter irrelevant sentences from threat reports, and extract TTPs from only relevant sentences to prepare TTP-based attribution dataset.

\section{CAPTAIN}
\label{sec:CAPTAIN}
% \begin{figure*}
%     \centering
%     % \includegraphics[width=\columnwidth, height = 19cm]{CAPTAIN_FINAL_1506-3.pdf}
%     \includegraphics[width=\textwidth]{fig/CAPTAIN.pdf}
%     \caption{CAPTAIN Architecture}
%     \label{fig:CAPTAIN}
% \end{figure*}

% \textbf{NOTATIONS AND THEIR DESCRIPTION Introduce tool}
% Our methodology CAPTAIN consists of three key components: TTP extraction (TTPXHunter), TTP sequencing, and attack pattern matching (Attribute APTs based on similar behaviour), shown in Fig.~\ref{fig:CAPTAIN}. In the TTP Extraction phase, we aim to extract TTPs from threat reports and create a dataset that pairs threat reports with their corresponding TTPs observed in the attack campaigns. Further, we arrange the TTP sequences based on the kill-chain phases of cyber-attacks, which signifies the sequential order of TTPs employed by attackers during the attack campaign. Next, we focus on capturing patterns represented by sequences of TTPs, which reflect the attacker's behavior throughout the attack stages. We prepare a database of TTP patterns observed in past attack campaigns. To attribute, we match similar behaviour by performing pattern matching using the proposed similarity measure. Within this process, we identify two significant aspects: the length of the matched pattern and the frequency of $n$-length matched patterns, which are crucial considerations during similarity matching, as they contribute to attributing the responsible APT group. We present detailed explanations of these components below:

Our methodology CAPTAIN consists of two key development phases: baseline establishment and similarity measure for attack pattern matching. In the baseline establishment phase, we prepare the TTP sequence database of APT groups observed in the past attack campaigns. We extract TTPs from the past attack campaign threat reports to prepare the database. Further, we arrange the extracted TTPs in sequences based on the kill-chain phases of cyber-attacks. It signifies the sequential order of TTPs employed by attackers during the attack campaign. In the next step, we develop a similarity measure to perform attack pattern matching, which aims to capture the patterns represented by sequences of TTPs reflecting the attacker's behaviour throughout the attack stages.
To attribute, we match similar behaviour by performing pattern matching using the proposed similarity measure. Within this process, we identify two significant aspects: the length of the matched pattern and their frequency, which are crucial considerations to identify similarity pattern, as they contribute to attributing the responsible APT group. We present detailed explanations of these components below, with a list of notations and their corresponding descriptions explained in Table~\ref{tab:parameters}.

\subsection{Baseline Establishment}
\label{subsec:DatabasePreparation}
To prepare the baseline database for attribution, we must have a comprehensive collection of TTPs used in past attack campaigns by APT groups. This collection represents the usual behaviour of known APT groups. No public repository exists for attack-wise TTPs, but security firms regularly publish reports detailing the tools and techniques used in past attack campaigns. We leverage these publicly available threat intelligence reports from blogs of leading security firms to address data challenges. We prepare a baseline database of TTP sequences seen in past attack campaigns. The database preparation consists of two steps: TTP extraction and TTP sequencing. The TTP extraction phase extracts TTPs from threat reports, and the TTP sequencing phase sequences the extracted TTPs based on the kill-chain phase, as shown in Fig.~\ref{fig:CaptainBaseline}. The details of both of these steps are discussed below.

\begin{figure*}
    \centering
    \includegraphics[width=\textwidth, height=11cm]{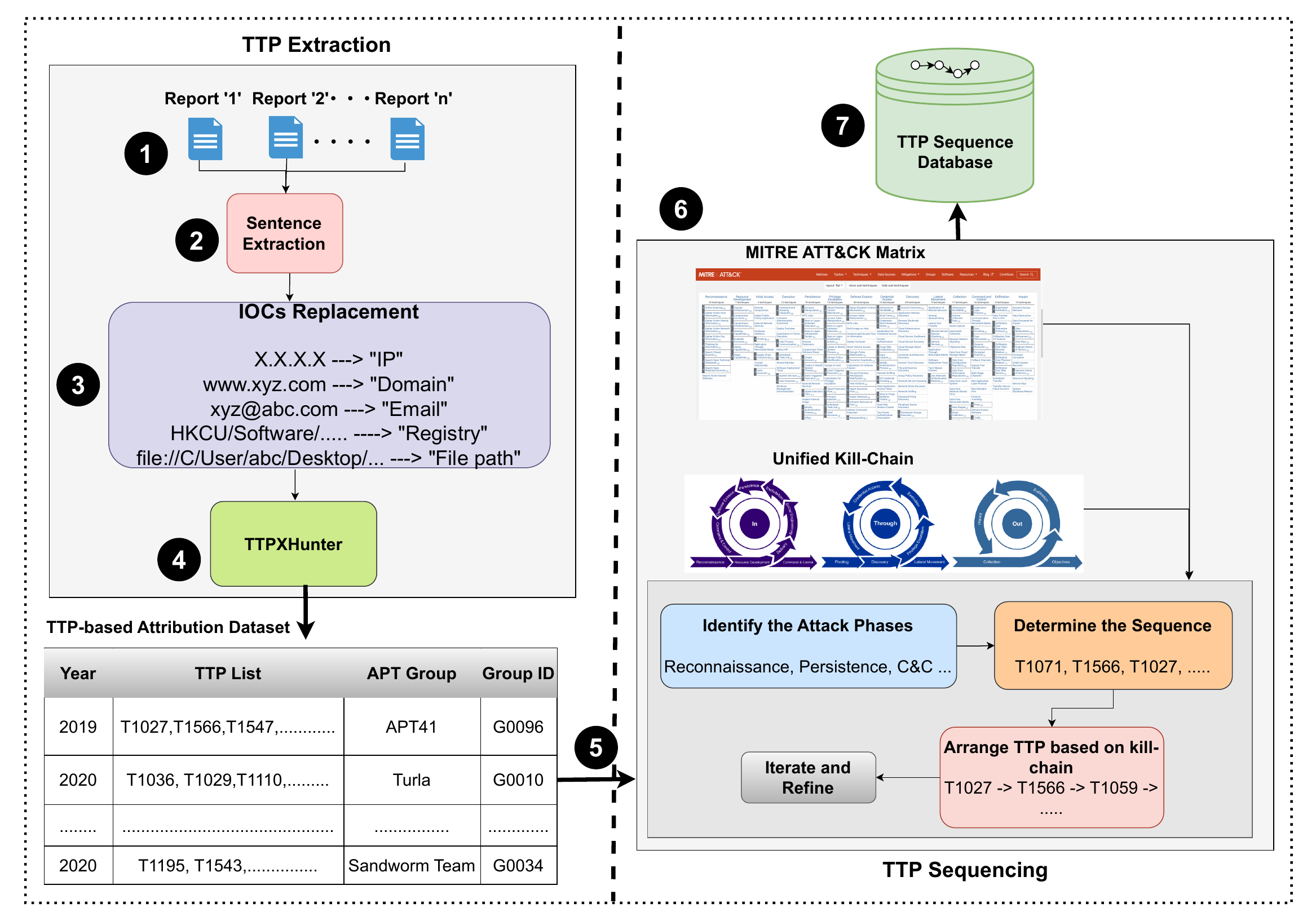}
    \caption{CAPTAIN: Baseline Establishment}
    \label{fig:CaptainBaseline}
\end{figure*}

\subsubsection{TTP Extraction}
\label{subsubsec:TTPExtraction}
% To attribute APTs by utilizing their TTPs, we must have a comprehensive collection of TTPs used in past attack campaigns. No repository or dataset is available to provide a collection of TTPs used by attackers in past attack campaigns. However, security firms regularly publish threat intelligence reports describing the attack campaign's details, including the tools and techniques used by the attackers. We leverage these publicly available threat intelligence reports from blogs of leading security firms to address data challenges. Although these 
The threat intelligence reports of security firms may not directly specify the TTPs used; they describe the TTPs in natural language. Therefore, we utilize  our tool named TTPXHunter~\cite{TTPXHunter2024Nanda} (discussed in Section~\ref{subsec:TTPXHunter}) to extract TTPs from these reports and collect the set of TTPs used by APT groups in past attack campaigns. 
We extract sentences from threat reports (step {\Large \textcircled{\normalsize 2}} in Fig.~\ref{fig:CaptainBaseline}) to identify TTPs explained in the sentences of reports.
We replace IOCs (Indicators of Compromise) with their base name present within the sentences in the dataset (step {\Large \textcircled{\normalsize 3}} in Fig.~\ref{fig:CaptainBaseline}). Specific patterns of IOCs present in the cyber-domain sentences, such as registry, IP address, domain name, and file path, can hinder the contextual interpretation of the sentences. To mitigate this issue, we identify a specific set of IOCs mentioned in~\cite{TTPXHunter2024Nanda,Rani2023TTPHunter} and employ a regex pattern to replace these patterns with their respective base name shown in Fig~\ref{fig:CaptainBaseline}.
We feed the processed sentences to our tool TTPXHunter to identify relevant sentences and recognize TTPs explained in the sentences (step {\Large \textcircled{\normalsize 4}} in Fig~\ref{fig:CaptainBaseline}). By employing this, we collect the TTPs seen in past attack campaigns. We label the TTPs set with the corresponding APT group name mentioned in the threat reports.

\subsubsection*{Assumptions}
The threat reports collected for this experiment go under the following assumptions:
\begin{itemize}
    \item We trust that the report has been prepared after completing a deeper analysis, which reflects the high quality and completeness of the threat report. To not compromise the report's trustworthiness, we select only prominent security firm's threat reports; the list is mentioned in Section~\ref{subsec:Dataset}.
    \item We trust the analysis and attribution performed by threat analysts of the prominent security firms who published the report. Therefore, we label the collected TTPs based on the threat group attributed in the corresponding threat reports.
    % \item For now, we used a state-of-the-art TTPXHunter for TTP extraction. In the future, if any more advanced state-of-the-art method comes, it can be replaced by that to extract TTPs.
\end{itemize}
% quality, completeness, trustworthiness, structure, what information is included, file format

% % \vspace{-1cm}
\subsubsection{TTP Sequencing}
\label{subsubsec:TTPsequence}
% The popular and traditional cyber kill chain (CKC)~\cite{Hutchins2011Intelligence} is malware-focused and fails to cover other attack vectors and attacks behind the organizational perimeter~\cite{Pols2017The}. Hence, we prefer unified kill-chain phases to prepare the TTP sequence framework~\cite{Pols2017The}. 
Once we extract TTPs from threat reports (step {\Large \textcircled{\normalsize 5}} in Fig.~\ref{fig:CaptainBaseline}), we perform the sequencing of extracted TTPs. In this phase, we aim to arrange extracted TTPs based on the flow of the kill chain. The kill-chain-based TTP sequence helps to understand the APT group's behavior, such as the set of methods a specific group generally implements to achieve a particular phase of attack. 
It also helps to understand and plot end-to-end attack flow. Our TTP sequence framework merges the UKC and MITRE ATT\&CK matrix (step {\Large \textcircled{\normalsize 6}} in Fig.~\ref{fig:CaptainBaseline}) and map the sequence of attacker's behaviour in terms of TTPs from reconnaissance to the final attack stage. We divide MITRE ATT\&CK TTPs and connect TTPs with the related UKC phase. 
The kill-chain phases are interconnected sequentially, illustrating the progression of an attacker's activities during an attack. Our framework leverages this interconnectedness to generate a complete TTP sequence that captures the attack flow. The key steps involved in creating the TTP sequence are as follows:
\begin{itemize}
    \item \textit{Identify the attack phases} - The first step categorizes the MITRE ATT\&CK TTPs into distinct groups or phases corresponding to different attack stages. These phases include reconnaissance, social engineering, pivoting, lateral movement, Exfiltration, and other stages defined by the UKC model. This categorization allows us to grasp the attacker's overall strategy and discern the specific objectives associated with each phase.
    \item \textit{Determine the sequence} - This step arranges the TTPs in a logical order based on UKC that reflects the typical progression of an attack. Starting with the reconnaissance phase, we progress through privilege escalation and command \& Control communication and continue through all stages until reaching the objectives stage. This sequential arrangement captures the natural flow of an attack.
    \item \textit{Iterate and refine} - This step involves reviewing and refining the attack flow with input from our security team. This feedback loop allows for incorporating any newly identified TTPs and refining the attack flow to align with changes in the evolving attack landscape.
\end{itemize}
% \subsubsection{Identify the attack phases} - The framework's first phase categorizes the MITRE ATT\&CK TTPs into distinct groups or phases corresponding to different attack stages. These phases include reconnaissance, social engineering, pivoting, lateral movement, Exfiltration, and other stages defined by the UKC model. This categorization allows us to grasp the attacker's overall strategy and discern the specific objectives associated with each phase.
% \subsubsection{Determine the sequence} - This phase arranges the TTPs in a logical order based on UKC that reflects the typical progression of an attack. Starting with the reconnaissance phase, we progress through privilege escalation and command \& Control communication and continue through all stages until reaching the objectives stage. This sequential arrangement captures the natural flow of an attack.
% \subsubsection{Iterate and refine} - The framework engages in an iterative process of reviewing and refining the flow graph with input from our security team. This feedback loop allows for incorporating any newly identified TTPs and refining the flow graph to align with changes in the evolving attack landscape. 

To understand how the TTP sequence is being created, we present an attack scenario as an example and demonstrate the creation of the TTP sequence. In the example attack scenario, the operation starts with Active Scanning (\texttt{T1595}) to perform reconnaissance and gather information about the victim's network. Next, the attacker focuses on Infrastructure Development (\texttt{T1583}) as the Resource Development phase for setting up command-and-control servers and creating malicious tools. Therefore, the sequence becomes \texttt{T1595 --> T1583}.
The attacker then sends spearphishing emails (T1566) to deliver and perform social engineering. Attackers use employee's actions to trigger Exploitation for Client Execution (\texttt{T1203}) and execute unauthorized code. Therefore, the sequence becomes \texttt{T1595 --> T1583 --> T1566 --> T1203}. After successful exploitation, the attacker establishes a network presence with the help of Process Injection (\texttt{T1055}) to gain persistence and uses the Obfuscated Files or Information (\texttt{T1027}) method to evade defense deployed at the victim side. Therefore, the sequence becomes \texttt{T1595 --> T1583 --> T1566 --> T1203 --> T1055 --> T1027}.
In the Command \& Control phase, attackers establish a communication channel using Application Layer Protocol (\texttt{T1071}) for c\&c communication, followed by performing Lateral Tool Transfer (\texttt{T1570}) in the Pivoting phase for network traversal. The attacker then engages in Account Discovery (\texttt{T1087}) and Exploitation for Privilege Escalation (\texttt{T1068}) to perform Discovery and Privilege Escalation, respectively.
The final stages involve Collection, Exfiltration, and Impact, marked by Automated Collection (\texttt{T1119}), Exfiltration Over the C2 Channel (\texttt{T1041}), and Data Destruction (\texttt{T1485}), completing the multi-stage attack. The complete attack TTP sequence utilizing our framework based on UKC for this attack scenario is shown in Fig.~\ref{fig:ttpseq}. 
\begin{figure}[!ht]
    \centering
    \includegraphics[width=\columnwidth, height=8cm]{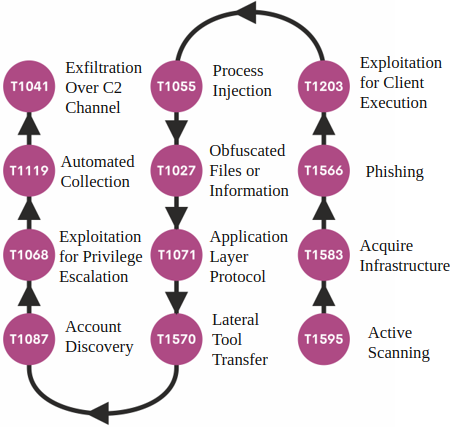}
    \caption{TTP sequence based on UKC}
    \label{fig:ttpseq}
\end{figure}
\\
\\
We generate a TTP sequence for each attack threat report, which explains the attack campaign in terms of TTP. We store the TTP sequences in a database (step {\Large \textcircled{\normalsize 7}} in Fig.~\ref{fig:CaptainBaseline}), which serves as a baseline database for the TTP sequence information, as shown in Fig.~\ref{fig:CaptainBaseline}. We leverage this database to analyze the sequence of TTPs used in an attack campaign. By examining these sequences, we identify behavioural patterns of APT groups across different attack campaigns.

\subsection{Similarity Measure}
\label{subsec:simialrityindex}
% Mathematical representation of similarity measure
Once we establish the baseline database, we perform attack pattern matching for any given TTP's sequence with established sequence database. CAPTAIN employ a similarity measure to perform the pattern matching and assign a similarity score for each comparison.
We propose a novel similarity measure to identify patterns in TTP sequences and attribute the associated APT groups based on similarity score achieved. The derivation and rationale behind this similarity index are detailed in Section~\ref{subsec:IntuttionProof}. The proposed similarity measure possesses essential properties such as non-negativity, boundedness, symmetry, and sensitivity to differences. The proof of these properties is in~\ref{sec:properties}.
The basis of our proposed similarity measure is the identification of common behaviour patterns and the frequency of such common patterns between two attack campaign's TTP sequences. By identifying matching patterns between sequences, we can infer similarities in the TTPs observed during the attack, which helps to attribute the responsible APT group. We extract the common matched patterns and evaluate the frequency of each matched pattern length, such as $2$-length, $3$-length, and till the longest subsequence length of the set of TTPs to quantify the degree of similarity between the sequences.

Consider two sequences, denoted as $seq_i$ and $seq_j$, having lengths $m$ and $n$, respectively. In the case where $m \geq n > 0$, we define the similarity score $Sim(seq_i, seq_j)$ as follows: \\

\begin{equation}\label{eq:1}
    \boxed{
    Sim(seq_i,seq_j) = \frac{2\langle \mu,\lambda \rangle} {m2^{m-1}+n2^{n-1}}
    }
\end{equation}
The variables $\lambda \in I^k$ is a length vector whose elements represent lengths of each common subsequence present in $seq_i \text{ and }seq_j$. The variable $\mu \in I^k$ is a frequency vector in which $i^{th}$ elements represent the number of occurrences of $\lambda_i$ length common subsequences. The $\langle \mu,\lambda \rangle$ represents the dot product operation of two vectors $\mu$ and $\lambda$.

\subsection{APT Attribution}
\label{subsec:APTAtribution}
% In the final phase of 
For the attribution in CAPTAIN, we use our proposed similarity measure to attribute the APT group responsible for a given attack campaign. Let $\mathcal{T}^{(g)}$ represent a sample of all TTP sequences belonging to group $g$, and let there be a total of $G$ different groups whose TTP sequences are present in the database. The CAPTAIN model employs the set of steps to attribute an attack campaign $A$ to a specific group within $1$ to $G$. The algorithm~\ref{alg:attribution} provides a pseudo-code representation of the proposed APT attribution method, and the steps are shown in Fig.~\ref{fig:CaptainPatternMatch}. The overall steps to perform APT attribution are the following:

\begin{algorithm*}[!ht]
\caption{APT Attribution algorithm}\label{alg:attribution}
\textbf{Input:} TTP sequences $(seq_A)$ \Comment{TTP sequence observed in given attack campaign $A$} \\
\textbf{Output:} Attributed APT group $(\hat{g})$
\begin{algorithmic}[1]
\Require Total number of groups $G$, captured TTP sequences $\mathcal{T}^{(g)}$ for every groups $g=1\; to\; G$
% \State $\mathcal{T}^{(g)} \gets$ ttp sequences of $g^{th}$ group
\State $m \gets \textnormal{length of } seq_A$
\For{$g=1$ to $G$}
    \State $SIZE[g] \gets \textnormal{number of TTP sequences in }g^{th}$ group 
\EndFor
% \State $n^{(g)} \gets \textnormal{number of sequences in } \mathcal{T}^{(g)}$
% \State $seq_A \gets test\_sequence$
\State $AVG\_SCORES \gets \textnormal{array of length } G$ \Comment{Initialize empty array to store attribution score for each group}
\For{$g=1$ to $G$}
    \State $SCORES \gets \textnormal{array of length } SIZE[g]$ \Comment{Initialize empty array to store similarity score within group $g$}
    % \State $\mathcal{T}_g \gets \textnormal{all samples of group } $g$ \textnormal{ in database}$
    \For{$i=1$ in $SIZE[g]$}
        % \State $seq \gets \mathcal{T}^{(g)}_i$
        \State $n \gets \textnormal{length of }\mathcal{T}^{(g)}_i$
        \State $\mu,\lambda \gets \Call{CSS}{seq_A, \mathcal{T}^{(g)}_i}$ \Comment{Calculate $\mu$ and $\lambda$ using $CSS$ function}
        \State $SCORES[i] \gets \frac{2\langle \mu,\lambda \rangle} {m2^{m-1}+n2^{n-1}}$ \Comment{Calculate similarity score from equation~\ref{eq:1}}
    \EndFor
    \State $AVG\_SCORES[g] \gets \frac{\sum SCORES}{SIZE[g]}$ \Comment{Average of similarity score within groups using equation~\ref{eq:2}}
\EndFor
\State \textbf{Return} $MAX(AVG\_SCORES)$ \Comment{Attribute to the group having maximum similar behaviour}
\vspace{0.2cm}
\Function{CSS}{$seq_A,seq_B$}
    \State $cs \gets $ set of all common subsequences present in $seq_A$ and $seq_B$
    % \State $k \gets$ lengh of $cs$
    \State $\lambda \gets \textnormal{ empty array}$
    \For{each $subseq$ in $cs$}
        \State $l \gets len(subseq)$
        \If{$l \notin \lambda$}
            \State $\lambda.append(l)$ \Comment{Length vector for elements of common subsequences}
        % \ElsIf{$l$ is in $\lambda$}
        %     \State $continue$
        \EndIf
    \EndFor
    \State $k \gets \textnormal{ length of } \lambda$
    \State $\mu \gets \textnormal{ zero vector of length } k$
    \For{$i=1$ to $k$}
        \State $l\gets \lambda[i]$
       \State $count \gets 0$
       \For{each $subseq$ in $cs$}
            \If{$l==len(subseq)$}
                \State $count \gets count+1$
            \EndIf
       \EndFor
        \State $\mu[i] \gets count$ \Comment{Frequency vector for elements of common subsequences}
    \EndFor
    \State \textbf{Return} $\mu,\lambda$
\EndFunction

\end{algorithmic}
\end{algorithm*}

\subsubsection{TTP Extraction}
\label{subsubsec:ttpextration}
This is the first step of the attribution method to collect the set of TTPs used in the attack campaign $A$. 
The incident response team of security firms analyses a range of various artifacts to understand the modus operandi used by threat actors. They structure them in terms of TTPs to list the set of TTPs observed. As we leverage the threat reports in this experiment, we extract the set of TTPs observed in the attack campaign $A$ as explained in section~\ref{subsubsec:TTPExtraction}.

\subsubsection{TTP Sequencing}
\label{subsubsec:ttpseq}
% CAPTAIN uses TTPXHunter to extract all TTP lists observed in the attack campaign $A$ as explained in section~\ref{subsubsec:TTPExtraction}. 
After the set of TTPs is observed, CAPTAIN prepares the sequence of TTPs. The TTPs sequence ($seq_A$) is prepared using the TTP sequencing framework explained in section~\ref{subsubsec:TTPsequence} as shown at step {\Large \textcircled{\normalsize 1}} in Fig.~\ref{fig:CaptainPatternMatch}. 

\subsubsection{Common Sub-Sequence (CSS) preparation}
\label{subsubsec:commonsubseq}
Once CAPTAIN obtains $seq_A$ for a given sample, it extracts all possible subsequences. It finds a common subsequence set (borrowed from~\cite{Wang2007A}) between $seq_A$ and sample present in the database ($\mathcal{T}_{i}^{(g)}$) as shown at step {\Large \textcircled{\normalsize 2}} in Fig.~\ref{fig:CaptainPatternMatch}. Where $\mathcal{T}_{i}^{(g)}$ represents $i^{th}$ sample of $g^{th}$ group.
Further, CAPTAIN calculates the length of a common subsequence (CSS) of a certain length; if any CSS exists, it stores the length in the $\lambda$ vector. Also, it computes frequency of all such length CSS and stores in frequency vector $\mu$ corresponding to the $\lambda$.

\begin{figure*}
    \centering
    \includegraphics[width=15cm, height=10cm]{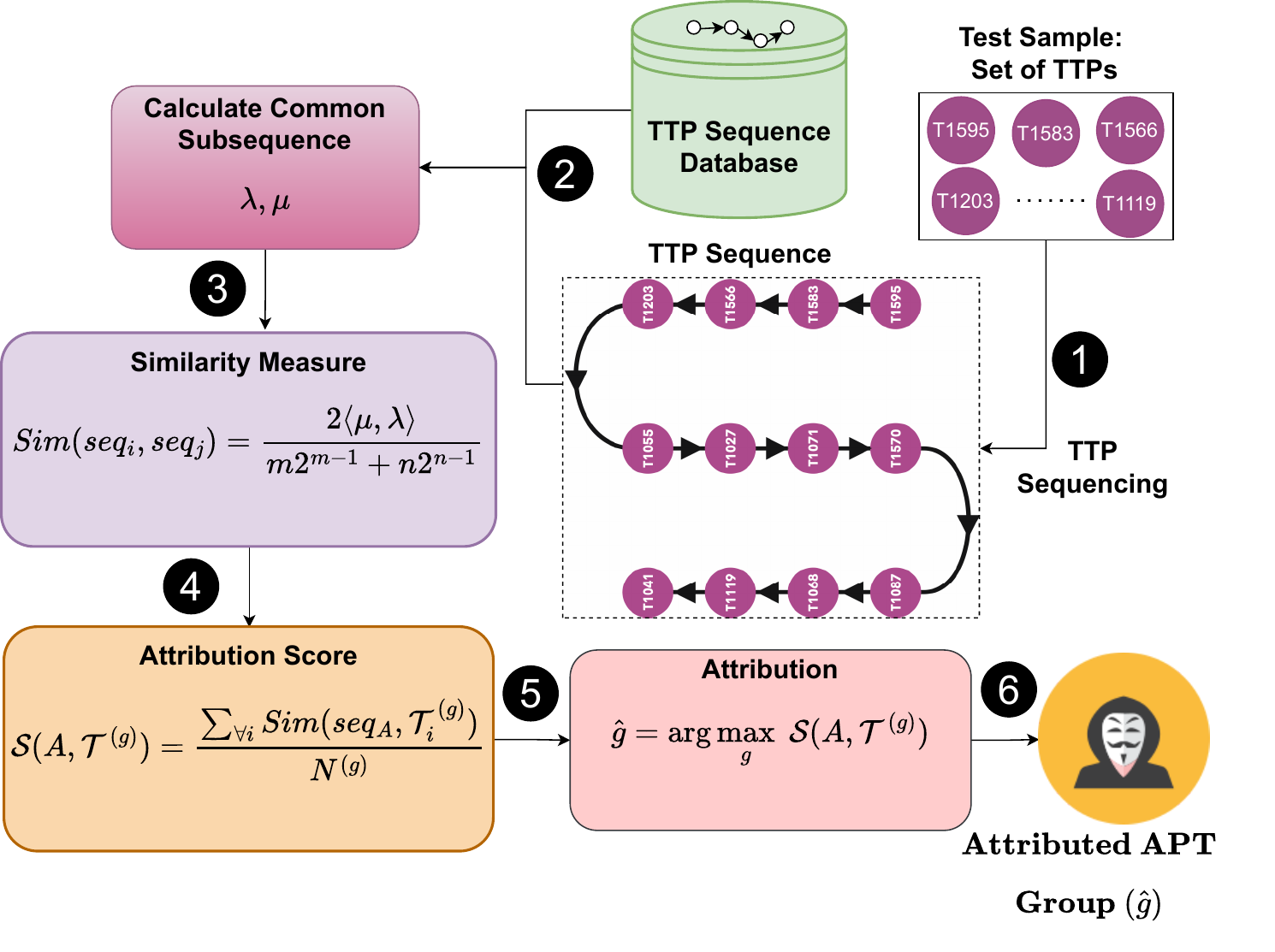}
    \caption{CAPTAIN: Pattern Matching}
    \label{fig:CaptainPatternMatch}
\end{figure*}

\subsubsection{Attribution Score}
\label{subsubsec:atributionScore}
Now, CAPTAIN feeds the $\mu$ and $\lambda$ to our proposed similarity measure (step {\Large \textcircled{\normalsize 3}} in Fig.~\ref{fig:CaptainPatternMatch}), present in section~\ref{subsec:simialrityindex}, and compute the attribution score ($\mathcal{S}$) of attack campaign $A$ for APT group $g$ (step {\Large \textcircled{\normalsize 4}} in Fig.~\ref{fig:CaptainPatternMatch}) as below:

\begin{equation}
\label{eq:2}
    \boxed{
    \mathcal{S}(A,\mathcal{T}^{(g)}) = \frac{\sum_{\forall i} Sim(seq_A,\mathcal{T}^{(g)}_i)}{N^{(g)}}
    }
\end{equation}
Where, $\mathcal{T}^{(g)}_i\in \mathcal{T}^{(g)}$, $N^{(g)}$ is number of TTP sequences in $\mathcal{T}^{(g)}$, and $seq_A$ is TTP sequence of target attack campaign $A$.

\subsubsection{Attribution}
\label{subsubssec:atribution}
Let we have TTP sequences for $G$ groups $(\mathcal{T}^{(1)},\mathcal{T}^{(2)},\cdots,\mathcal{T}^{(G)})$ in the database. Once we compute the attribution score ($\mathcal{S}$) for all known APT groups, the group having the maximum score (step {\Large \textcircled{\normalsize 5}} in Fig.~\ref{fig:CaptainPatternMatch})) shows high confidence in getting attributed. Therefore, CAPTAIN attributes the group ($\hat{g}$), shown at step {\Large \textcircled{\normalsize 6}} in Fig.~\ref{fig:CaptainPatternMatch}, which has maximum score as below:

\begin{equation}\label{eq:ettribution}
\boxed{
\hat{g} = \arg\max_{g}\;\mathcal{S}(A,\mathcal{T}^{(g)})
}
\end{equation}

\subsection{Intuition and Properties}
\label{subsec:IntuttionProof}
The proposed similarity measure satisfies non-negativity, boundedness, symmetry, and sensitivity to difference properties. The non-negativity properties show that the measure is always a non-negative value; boundedness represents the similarity measure between the lower bound as $0$ and the upper bound as $1$. The symmetry property represents that $Sim(seq_a,seq_b)$ is the same for $Sim(seq_b,seq_a)$ and sensitivity to difference illustrates how the similarity measure gets affected if the TTP of common subsequence changes to uncommon or common TTPs. We discuss four possible cases to demonstrate this property for the proposed similarity measure.
The proof of these properties is present in~\ref{sec:properties}.

We develop the proposed similarity measure by incorporating the insights from the Longest Common Subsequence (LCS)~\cite{Hirschberg1977Algorithms} and Gelstat Pattern Matching~\cite{Ratcliff1988Pattern} algorithms. The LCS algorithm contributes to calculating subsequences in a given input sequence, while the Gelstat method enables similarity measurement. By leveraging the subsequence analysis from LCS and the pattern-matching approach from Gelstat, we enhance the effectiveness of our similarity measure in identifying and matching patterns.
Identifying similarities in an attacker's behavior between two attack campaigns relies on extracting properties from common subsequences, which aids in recognizing behaviour patterns. In this context, we identify the key properties, including the length of matched patterns and their frequency within the sequence.
% We perform several experiments, such as matching the attacker's behaviour without these properties and including these properties. We find that length and frequency contribute significantly to identifying the similarity in the attacker's behavior between attack campaigns. Our findings demonstrate that incorporating length and frequency enhances the ability to identify similarities in the attacker's modus operandi between attack campaigns.
\subsubsection{Derivation of similarity measure}
\label{subsec:Derivation}
Once we obtain a common subsequence between all subsequences of $seq_i$ and $seq_j$, we calculate the contribution of each common subsequence in the similarity measure. The proposed similarity measure is based on \textbf{how long} and \textbf{how many} common subsequences are present. Therefore, we multiply the length of each common subsequence $(\lambda_i)$ with the frequency of such length patterns $(\mu_i)$ to calculate the weightage of each subsequence in the similarity matching as follows.
\begin{align*}
    Sim'(seq_i,seq_j) &= \sum_{i=0}^{len(\lambda)} \lambda_i \mu_i\\
    &= \langle \mu,\lambda \rangle
\end{align*}

A similarity measure is a measure that is relative to the total weightage of the highest possible similarity index. Therefore, we divide the above index with a measure of the total weight of the CSS similarity index, including $seq_i$ and $seq_j$. Further, we multiply the above weight with $2$ because the common subsequence is present in both given input sequences and aids in normalizing the measure~\cite{Ratcliff1988Pattern}. Now, the above formula becomes as follows:

\begin{equation}
    \label{eq:4}
    Sim(seq_i,seq_j) = \frac{2\langle \mu,\lambda \rangle}{\sum_{i=0}^{i=m} i \binom{m}{i} + \sum_{i=0}^{i=n} i \binom{n}{i}}
\end{equation}
In the denominator, $\binom{n}{i}$ represents the total possible subsequences of length $i$ in a given $n$-length sequence. We further simplify the above equation in the compact form using claim~\ref{claim1}. The following equation is a simplified form of the similarity measure, also mentioned in equation~\ref{eq:1}.

\begin{equation*}
    \boxed{
    Sim(seq_i,seq_j) = \frac{2\langle \mu,\lambda \rangle} {m2^{m-1}+n2^{n-1}}
    }
\end{equation*}

\begin{claim}
\label{claim1}
The following formula holds true:
$$\sum_{i=0}^{i=n} i \binom{n}{i} = n2^{n-1}$$
\end{claim}
\begin{proof}\renewcommand{\qedsymbol}{}
Consider the binomial theorem:
\begin{align*}
    \sum_{i=0}^{n} x^i\binom{n}{i} &= (1 + x)^n\\
    \frac{d}{dx}\sum_{i=0}^{n} x^i\binom{n}{i} &= \frac{d}{dx}(1 + x)^n && \text{\hspace{0.2cm}differentiation wrt. $x$}\\
    \sum_{i=0}^{n} ix^{i-1}\binom{n}{i} &= n(1 + x)^{n-1}\\
    \sum_{i=0}^{i=n} i \binom{n}{i} &= n2^{n-1} && \text{\hspace{0.2cm}replace $x=1$}
\end{align*}
\end{proof}
% \vspace{-1.235cm}

\subsection{Computation Complexity}
\label{subsec:complexity}
We implement CAPTAIN using the Algorithm~\ref{alg:attribution}. The main computation cost of CAPTAIN lies in two loops which calculates all pair similarity scores (line 6 and line 8 in Algorithm~\ref{alg:attribution}). The total iteration of both loops is the number of TTP sequences for every group. Let us assume the baseline dataset has a total $N$ number of TTP sequences. Then, the computation cost of CAPTAIN is $N$ times the computation cost of the body of the loops. In the iteration body, the two major computations are computing common subsequences 
(line 10 in Algorithm~\ref{alg:attribution}) 
and the dot product of $\mu$ and $\lambda$ vectors (line 11 in Algorithm~\ref{alg:attribution} ). Getting a set of common subsequences (CSS) dominates the computation cost. The complexity of function $CSS(seq_i,seq_j)$ is $\mathcal{O}(m\times n)$~\cite{Wang2007A,ACSGitHub}, where $m$  and $n$ are the lengths of $seq_i$ and $seq_j$, respectively. In the case of the dot product, the maximum length of $\mu$ and $\lambda$ vectors can be $MIN(m,n)$. The complexity of performing dot product is $\mathcal{O}(MIN(m,n)^2)$. Now, we can combine all the above-discussed complexities as $\mathcal{O}(N\times (\mathcal{O}(m\times n)+ \mathcal{O}(MIN(m,n)^2)))=\mathcal{O}(N\times m\times n)$.

In the proposed APT attribution method, we prepare TTP sequences for the $18$ attack stages of the UKC. Based on domain knowledge, mostly one or two techniques are used to perform a tactic. Also, the attacker does not need to perform all the $18$ stages to execute an attack~\cite{Pols2017The}.
Based on the above discussion, we can say that a TTP sequence is limited in size. Also, in our dataset of $580$ TTP sequences (shown in Fig.~\ref{fig:TTPNumberDist}), we find that the average length of a TTP sequence is $12.22$ (minimum=4, maximum=34, and std=6.33). Therefore, we can consider that the value of $m$ and $n$ are constant and $\mathcal{O}(N\times m\times n)=\mathcal{O}(N)$, which leads to the time complexity of CAPTAIN is $\mathcal{O}(N)$.

CAPTAIN needs to store every group's TTP sequence while attributing a TTP sequence that requires $\mathcal{O}(N)$ space complexity. As discussed, other variables used during the computation process are smaller, temporary, or constant. Therefore, the space complexity of CAPTAIN is $\mathcal{O}(N)$.
\begin{figure}[!ht]
    \centering
    \includegraphics[width=\columnwidth]{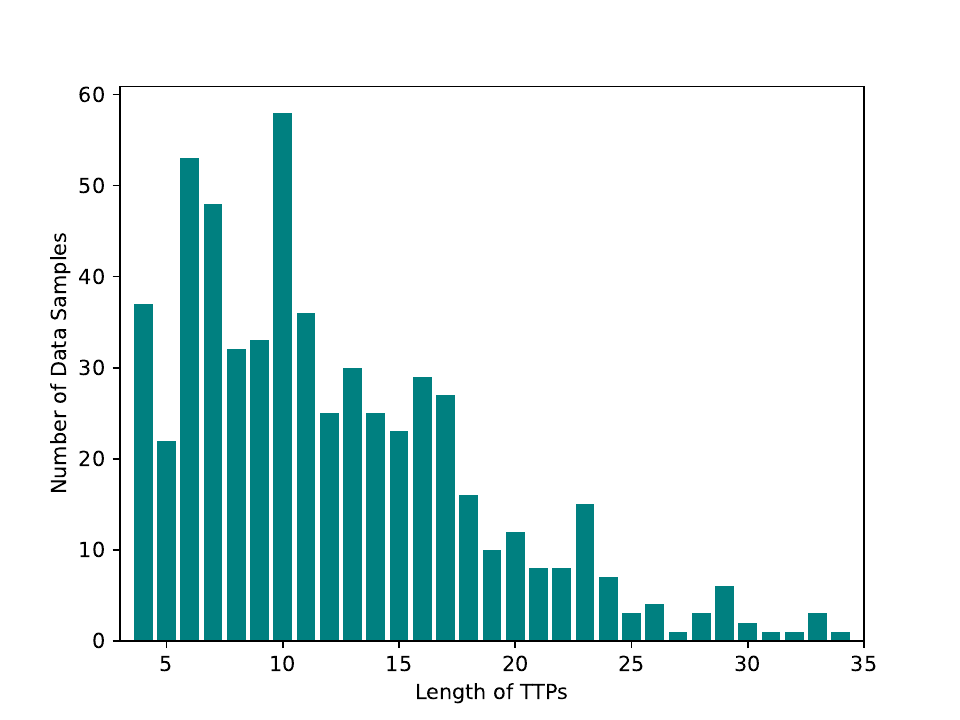}
    \caption{Number of data samples consisting various length of TTP sequence in TTP-APT dataset. Mean = $12.22$ and std = $6.33$}
    \label{fig:TTPNumberDist}
\end{figure}

\section{Experiments \& Results}
\label{sec:ExperimentResult}

We perform experiments to implement and evaluate proposed method CAPTAIN. We compare CAPTAIN's performance with available literature and traditional similarity measures. The experiment and evaluation aim to employ reasoning over obtained results, and this section explains them in detail.
% address presented research questions.
% \fix{Add Content}

\subsection{Performance Metrics}
\label{subsec:PerformanceMetrics}
We assess the performance of implemented methods using macro metrics, i.e., macro precision, macro recall, and macro f1-score. Macro metrics provide equal weights to every threat group class while averaging overall performance metrics, which prevents any biasness introduced by majority target class~\cite{grandini2020metrics}.
% The TTP-based attribution dataset is imbalanced, as shown in Table.~\ref{tab:ttpaptdistribution}, it may introduces biasness when evaluating models. Therefore, we calculate macro metrics, i.e., macro precision, macro recall, and macro f1-score, to provide equal weights to every threat group class while averaging overall performance metrics~\cite{grandini2020metrics}. 
In the context of attributing cyber threats, the goal is to get better at attributing correctly linked threat groups. Therefore, better attribution methods should focus on minimizing false positives over false negatives, i.e., improving precision~\cite{goutte2005probabilistic}. Emphasizing precision helps to make confident and informed decisions, ensuring that resources are directed appropriately based on precise attributions. Consequently, we consider precision a high-priority performance metric while comparing attribution performance.

\subsection{Dataset Collection}
\label{subsec:Dataset}

We collect $580$ threat reports published by various prominent security firms~\footnote{From various security firm's blogs: \href{https://www.zscaler.com}{Zscaler Inc.}, \href{https://www.checkpoint.com}{Check Point Software Technologies Ltd.}, \href{https://www.forcepoint.com}{Forcepoint}, \href{https://www.pwc.com}{PWC}, \href{https://www.proofpoint.com}{Proofpoint Inc.}, \href{https://www.clearskysec.com}{Clearskysec}, \href{https://www.crowdstrike.com}{CrowdStrike Holdings Inc.}, \href{https://www.trendmicro.com/}{Trend Micro Inc.}, \href{https://symantec-enterprise-blogs.security.com/blogs/}{Symantec (Now part of Gen Digital Inc.)}, \href{https://www.kaspersky.com}{Kaspersky Lab}, \href{https://www.paloaltonetworks.com}{Palo Alto Networks Inc.}, \href{https://blogs.cisco.com/}{Cisco Systems Inc.}, \href{https://www.welivesecurity.com}{WeLiveSecurity}, \href{https://www.trellix.com/en-us/index.html}{Trellix}} to prepare a TTP-based attribution dataset. We employ our tool TTPXHunter, an extended version of our previous tool named TTPHunter~\cite{Rani2023TTPHunter}, to extract TTPs (step {\Large \textcircled{\normalsize 4}} in Fig.~\ref{fig:CaptainBaseline}) for attack campaigns from collected threat reports.
Our dataset consists of $580$ TTP sequences for $11$ APT groups. The prepared dataset contains seven columns: Year (in which year the report was published), TTPs (list of TTPs observed in the attack campaign), APT Group (attributed APT group name), and group ID (ID given by MITRE~\footnote{\url{https://attack.mitre.org/groups/}}), Group Aliases (given by MITRE), File Name (name of the threat report), and Report link (webpage address of threat report). On average, the number of samples in our dataset is $52$. The group $APT17$ consists of the highest samples at $94$, and the group $APT3$ consists of the lowest samples at $9$.
The distribution of data samples of each group is shown in Table~\ref{tab:ttpaptdistribution} and a glimpse of the dataset is shown at step {\Large \textcircled{\normalsize 5}} in Fig~\ref{fig:CaptainBaseline}.
% We keep a total of $435$ number of samples in the database and $145$ number of samples for evaluation.

\begin{table}[!ht]
    \centering
    \caption{Distribution of TTP-based Attribution Dataset}
    \begin{threeparttable}
   % \begin{tabular}{|c|c|c|}
%     \hline
%     \textbf{APT Group} & \textbf{No. of Samples} & \textbf{Linked Country} \\
%     \hline
%     \hline
%     OilRig & $69$ & Iran \\
%     \hline
%     APT17 & $94$ & China \\
%     \hline
%     Lazarus Group & $37$ & North Korea \\
%     \hline
%     Turla & $75$ & Russia \\
%     \hline
%     APT41 & $25$ & China \\
%     \hline
%     RocketKitten & $12$ & Iran \\
%     \hline
%     FIN7 & $81$ & Unknown \\
%     \hline
%     APT29 & $88$ & Russia \\
%     \hline
%     APT3 & $9$ & China \\
%     \hline
%     APT28 & $47$ & Russia \\
%     \hline
%     APT10 & $43$ & China \\
%     \hline
% \end{tabular}

\begin{tabular}{ccc}
    \hline
    \textbf{Group ID\tnote{*}} & \textbf{APT Group} & \textbf{No. of Samples} \\
    \hline
    % \hline
    G0049 & OilRig & $69$ \\
    % \hline
    G0025 & APT17 & $94$ \\
    % \hline
    G0032 & Lazarus Group & $37$ \\
    % \hline
    G0010 & Turla & $75$ \\
    % \hline
    G0096 & APT41 & $25$ \\
    % \hline
    G0130 & RocketKitten & $12$ \\
    % \hline
    G0046 & FIN7 & $81$ \\
    % \hline
    G0016 & APT29 & $88$ \\
    % \hline
    G0022 & APT3 & $09$ \\
    % \hline
    G0007 & APT28 & $47$ \\
    % \hline
    G0045 & APT10 & $43$ \\
    \hline
\end{tabular}
   \begin{tablenotes}
       \item[*] APT group id given by MITRE 
   \end{tablenotes}
   \end{threeparttable}
    \label{tab:ttpaptdistribution}
\end{table}

% \subsection{Baseline Establishment}
% \label{subsec:BaselineEstab}
% To establish the baseline database, we select $435$ samples out of $580$ samples collected in dataset collection phase (Discussed in Section~\ref{subsec:Dataset}). We prepare TTP sequence of each TTP samples present in database by following method discussed in Section~\ref{subsubsec:ttpseq}. The prepared sequence database is an established baseline of TTP sequences.

\subsection{CAPTAIN Implementation}
\label{subsec:CaptainImplem}
The first step to implment CAPTAIN is to establish a baseline database. To establish the baseline database, we select $435$ samples out of $580$ samples present in the dataset collection. We prepare the TTP sequence of each TTP sample present in the database by following the method discussed in Section~\ref{subsubsec:TTPsequence}. We grouped the TTP sequences based on their threat group name label. The prepared group-wise sequence database is an established baseline of TTP sequences.

Once we established the database, we considered the remaining samples, i.e., $145$ samples, as an evaluation set to evaluate the performance of implemented methods. We prepare TTP sequences of evaluation set samples using the method discussed in Section~\ref{subsubsec:TTPsequence}. After preparing evaluation set TTP sequences, we compare them with the established database using the proposed similarity measure discussed in Section~\ref{subsec:simialrityindex}. We follow Algorithm~\ref{alg:attribution} for each comparison between the evaluation set and the established database and later steps to perform the attribution. Once attribution is performed for each sample, we evaluate the performance of CAPTAIN using metrics discussed in Section~\ref{subsec:PerformanceMetrics}. The CAPTAIN obtained a precision of $61.36\%$, and the result is shown in Table~\ref{tab:attributionresult}.

\begin{table}[ht]
    \centering
    \caption{Attribution Performance Comparison (in \%)}
    \label{tab:attributionresult}
    \begin{tabular}{ccccc}
    \hline
     \textbf{Method} & \textbf{Accuracy} & \textbf{Precision} & \textbf{Recall} & \textbf{F1-score} \\
     \hline
     % \hline
     Cosine & $33.10$ & $58.17$ & $43.23$ & $34.49$ \\
     % \hline
     Euclidean & $40.0$ & $57.23$ & $53.65$ & $38.73$ \\
     % \hline
     LCS & $33.10$ & $56.63$ & $44.25$ & $34.30$ \\
     % \hline
     \textbf{CAPTAIN} & \textbf{55.17} & \textbf{61.36} & \textbf{64.74} & \textbf{52.72} \\
     \hline
\end{tabular}
\end{table}

\subsection{Comparison with Traditional Similarity Measures}
\label{subsec:CompareTraditionalSim}
In this experiment, we aim to compare our proposed similarity measure with traditional similarity measure to assess its effectiveness in performing attribution. We compare attribution performance by replacing our proposed similarity measure in CAPTAIN with traditional similarity measures. We consider three traditional measures for this experiment: Cosine, Euclidean, and Longest Common Subsequence (LCS). Cosine and Euclidean measures distinguish threat actor's behaviour based on mere presence on TTPs, Whereas LCS leverages sequence property to determine similarity index. The LCS similarity method works by comparing the elements of the two sequences and finding the longest common subsequence~\cite{Bakkelund2009An}.

We leverage CAPTAIN's established baseline, evaluation dataset, and performance metrics to implement and compare experiments on the common ground. To add the presence or absence of TTPs in the given data sample, we leverage the one-hot encoding method and prepare an input feature vector for the Cosine and Euclidean similarity measure. Next, the encoded input vector is given to Cosine and Euclidean measure to compare, assign attribution score and attribute the most likely threat group similar to CAPTAIN implementation discussed in Section~\ref{subsec:CaptainImplem}. We individually implement these three traditional similarity measures with CAPTAIN to perform attribution and evaluate their corresponding performance. We calculate attribution performance based on all these three traditional measures, and the results are shown in Table~\ref{tab:attributionresult}. As precision is more valuable than other performance metrics, we consider the precision score as the primary comparison metric. It is noticeable that CAPTAIN, based on our proposed similarity measure, performs better than traditional similarity measures. Our proposed method, CAPTAIN, provides the highest precise attribution model with a $61.36\%$ score. Cosine demonstrates better results among all other traditional similarity measure-based implemented methods, i.e., $58.17\%$ precision score, but is low compared to the proposed CAPTAIN method. It demonstrates that the properties, including lengths of pattern sequences and their corresponding frequencies, significantly enhance the capability to identify similarities in threat actor's behaviour.

% \begin{table*}[!ht]
% \caption{Attribution Performance Comparison}
%     \label{tab:attributionresult}
%     \centering
%     \input{tables/AttributionPerformance}
% \end{table*}

\subsection{Comparison with Literature}
In our literature review, we find Noor et al.~\cite{Noor2019A}, which presents a method for TTP-based APT attribution. This method is based on only the existence of TTPs and implements various ML models such as KNN (K-Nearest Neighbors), Decision tree, Random Forest, Naive Bias, and DLNN (Deep Learning Neural Network). They evaluate all implemented models and find that DLNN outperforms other ML models. Therefore, we implement DLNN to assess the performance of their model for our dataset in this experiment. We consider established database as training set and evaluation set as test set to compare it with CAPTAIN on common ground. The implemented DLNN consists of 2 layers of 64 neurons each. The hidden layer relies on relu activation functions, and the classification layer consists of a softmax activation function.

A comparison between CAPTAIN and the implemented method is shown in Fig.~\ref{fig:CaptainNoor}. 
We can see that the literature which is a training-based method and rely on mere presence of TTPs is ineffective in identifying patterns of attacker behaviour for attribution. 
\begin{figure}[h]
    \centering
    \includegraphics[width=\columnwidth]{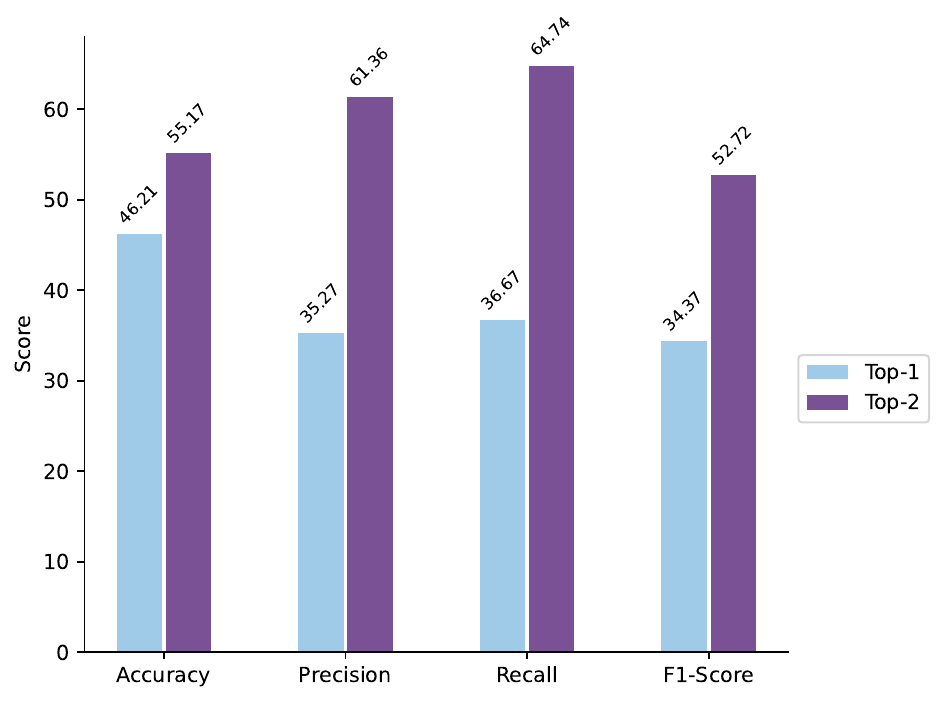}
    \caption{Performance Comparison between CAPTAIN and Noor et al.~\cite{Noor2019A}}
    \label{fig:CaptainNoor}
\end{figure}
The reason can be that the model could not understand the patterns because of limited data per threat group. As discussed in Section~\ref{sec:introduction}, getting a substantial number of attack campaigns for any threat group is challenging. The result in Fig.~\ref{fig:CaptainNoor} demonstrates that pattern matching-based our proposed method CAPTAIN is performing well in such cases. 
Overall, the results highlight the effectiveness of CAPTAIN in accurately attributing attacker behaviour patterns, demonstrating its superiority over traditional approaches.

\subsection{Attribution Certainty}
Attribution cannot always be certain because of its probabilistic nature and false flags planted by sophisticated actors~\cite{Steffens2020Attribution,lindsay2015tipping,Skopik2020Under}. 
Also, given the inherent challenge of limited data in the attribution problem, we adopt a relaxed evaluation approach inspired by~\cite{Rosenblum2011Who,Almasian2022Bert,Chen2019Chained}. 
According to the relaxed measurement, it is a good practice to narrow down the probably responsible actor list rather than attribute it to only one group. 
% In the relaxed measurement, we consider the model's output correct if the true output falls within the top-K predicted output set.
% For the attribution problem, we evaluate the proposed model's performance using top-1 and top-2 attributed APT groups and report performance in both cases. 
By considering probable predictions, we gain a broader understanding of the model's attribution capabilities and ability to identify potential APT groups involved in the attack.
Therefore, we evaluate the proposed model's performance using the top-2 method to narrow down attribution to two possible threat groups. Further, based on other technical and non-technical artifacts, an analyst can pin down the most aligned threat group as an attributed group. The main objective of the top-2 method is to narrow down the analysis range from many to two and minimize the manual analysis burden from the analyst's shoulder. In the top-2 approach, we consider attribution correct if the true value is present in the topmost two attributed groups; otherwise, it is false. The performance of the top-2 experiment over top-1 is shown in Fig.~\ref{fig:Top1Top2Performance}.
\begin{figure}[h]
    \centering
    \includegraphics[width=\columnwidth]{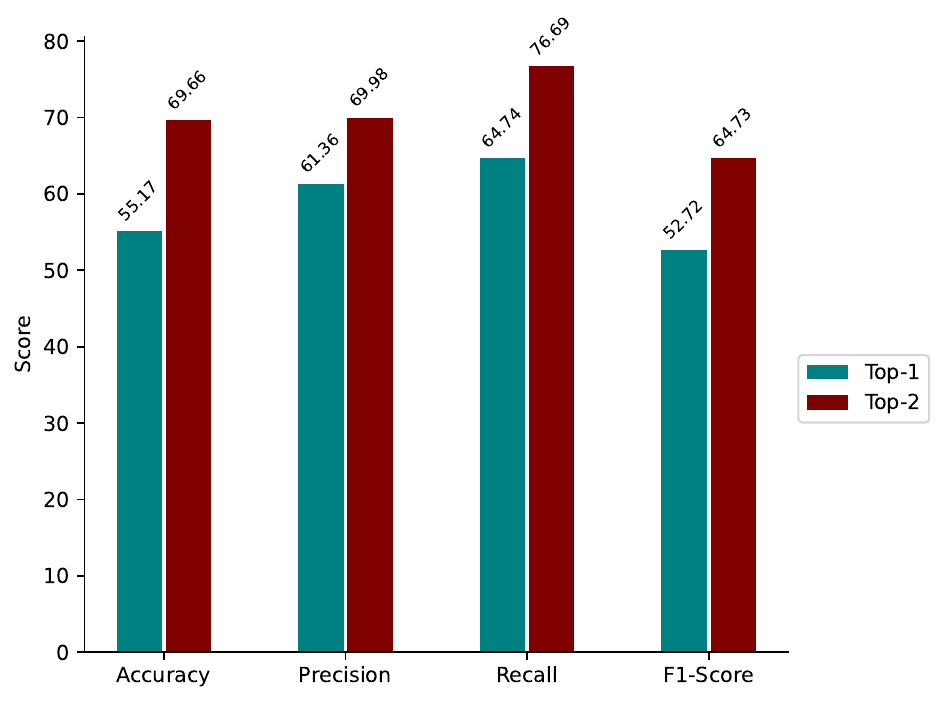}
    \caption{Top-1 and Top-2 Performance}
    \label{fig:Top1Top2Performance}
\end{figure}
We can see significant improvement in performance after introducing a window of two to narrow down the probably responsible threat group. Further, we extend the top-2 experiment to top-n to analyze the better value for $n$ if one wants a better precise model. Therefore, we evaluate the precision score for each possible value of $n$, i.e., $1$ to $n$, and the result is shown in Fig.~\ref{fig:PrecisionVsTopn}. Taking the top-3 method, we observe that the attribution reaches a precision of $74.24\%$, and it further improves with the expansion of the top-n window.
Further, an analyst can use other indicators or artifacts used in the operation to identify the attacker out of the list of narrowed threat actors. The narrowed list enables analysts to identify which threat group they may be dealing with. Having a narrowed list of probable threat groups significantly reduces the analyst's effort, as they can now focus on a limited number of threat groups instead of a wide variety.
% Narrowing down the probable threat group list can significantly reduce the effort of analysts from analysing various threat groups to narrowed threat groups.
% Further, analysts are required to conduct analyses using other artifacts or IOCs from the identified APT groups to discover related threat groups.
% FIXFIX \fix{Add about values 1-2 lines}
\begin{figure}[!ht]
    \centering
    \includegraphics[width=\columnwidth]{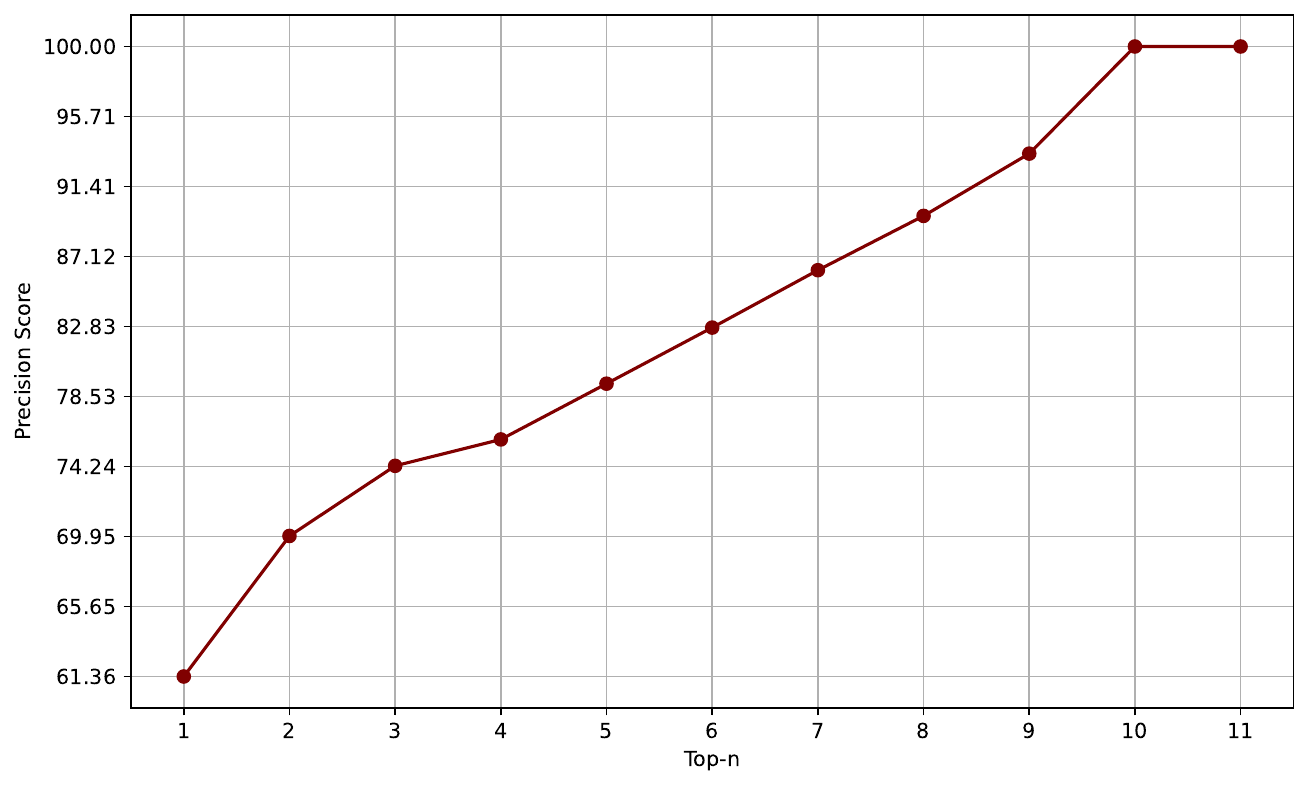}
    \caption{Precision vs Top-n Analysis}
    \label{fig:PrecisionVsTopn}
\end{figure}

\subsection{Sequence Correlation}
\label{subsec:SeqCorr}
We analyze the correlation between the APT groups present in the dataset by using our proposed similarity measure. This analysis aims to validate the TTP sequencing. We obtained comprehensive insights by calculating the similarity measure for all possible pairs in our dataset and then averaging the intra- and inter-group scores. The results are depicted in a heatmap in Fig.~\ref{fig:SequenceCorrelartionHeatmap}. Interestingly, we observe a relatively high correlation between the same-linked country APT17 and APT41 groups~\cite{MITREgroup}, which shows that the sequence between both groups is comparatively correlated. One recent research~\cite{secalliance2023} demonstrates that these groups are linked, which can be the reason for the higher correlation between both group's modus operandi. Beyond this observation, we noted a stronger correlation among intra-groups than inter-group samples. This pattern reflects that the UKC-based sequencing of modus operandi introduces contextual depth to the sequence and enriches the understanding of the group's operational behaviors.
\begin{figure}[ht]
    \centering
    \includegraphics[width=\columnwidth, height=6.5cm]{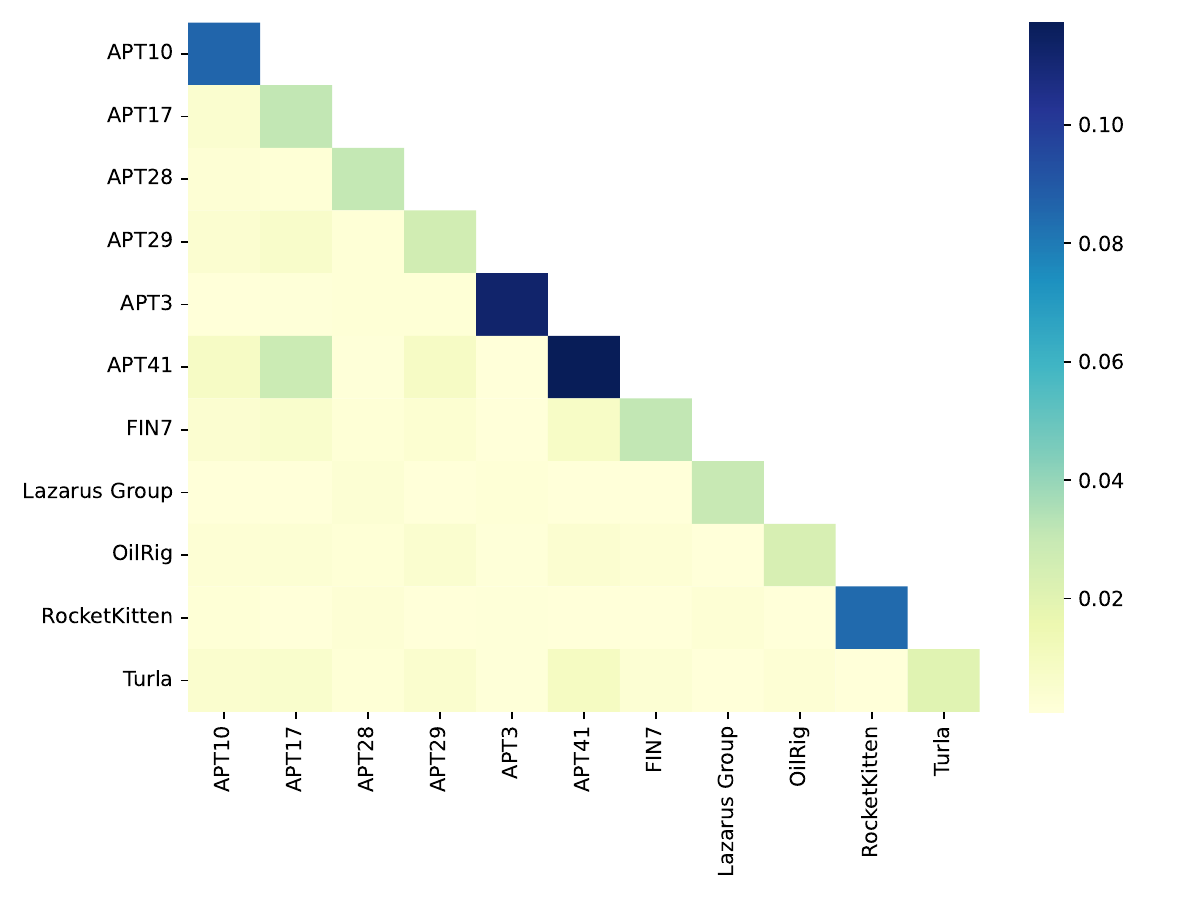}
    \caption{Correlation Measure between APT Groups}
    \label{fig:SequenceCorrelartionHeatmap}
\end{figure}
\\
\\
We find that TTPs are valuable for distinguishing between different attacker's behaviours based on the techniques employed at specific attack stages. The performance of CAPTAIN supports the valuation of TTPs in attribution, which starts with the foundational concept of the Pyramid of Pain and demonstrates the importance of TTPs in identifying attacker's repetitive behaviour.
In addition to the presence of TTPs, we observe that merely the presence of TTPs contributes less than the properties of TTP sequences added to our similarity measure.
The propsoed similarity measure considers the length of matched behaviour patterns (i.e., the length of matched TTP subsequences) and the frequency of such patterns within an attack campaign. By incorporating these factors, we establish connections between the behaviour patterns and the responsible entity.
Further, the result of literature shows that the training-based and mere presence of TTPs-based method is inefficient in performing threat attribution compared to the proposed method CAPTAIN. 
Overall, our findings shed light on the significance of TTPs and the value of our proposed similarity measure in capturing and analyzing attacker behavior, thus contributing to cyber threat attribution.

\section{Challenges \& Limitations}
\label{sec:ChallengeLimitation}
% Challenging to acquire attack campaign data
% Unable to employ ML methods as couldn't collected the data much
% False Flag in threat attribution
% Attacker's modus operandi evolves over time and need to acknowledge
% Data Bias like security firm who analysed and labelled the attack campaign we need to trust on it and can be biased based on security firm's geolocation and interest
% We face significant challenges in APT attribution, primarily in acquiring the attack campaign data. This scarcity of data hampers our ability to employ machine learning methods effectively, as these approaches require extensive datasets to train and validate models. Additionally, t
The prevalence of false flag tactics in cyber attacks, where attackers deliberately mislead about their identity, further complicates accurate attribution. Compounding the issue of the evolving nature of the attacker's modus operandi demands continuous updates and adaptations in our analysis methods. Furthermore, there can be an inherent data bias from relying on security firms to analyze and label attack campaigns. The geographical location and interests of these firms can influence this bias. Collectively, these factors underscore the complex and dynamic challenges inherent in attributing cyber threats and open doors for further research in this domain.

The sequencing may sometimes need manual intervention to ensure contextual behaviour sequencing, specifically in two cases: more than a TTP extracted for specific tactics, and the extracted TTP belongs to more than one tactic. For both cases, we manually read the reports and adjust the sequencing based on the way TTPs are executed by the threat actor. The ATT\&CK framework consists of a few TTPs in multiple tactics. We adjust such TTPs based on the tactics by which the threat actor uses the TTPs in the specific attack campaign. The number of such TTPs is small in comparison, i.e., $27$ out of $201$ TTPs in the current framework version. This manual intervention is minimal compared to framework-based analysis. The list of such TTPs is present in~\ref{sec:TTPs}.
% \fix{Mention manual intervention in TTP sequencing}

\section{Conclusion}
\label{sec:conclusion}
% 1. Conclude contribution\\
This research finds the importance of TTPs in attributing APT groups and their characteristics that we need to take care of while performing attribution. The proposed framework CAPTAIN finds that the sequence of TTPs (how long) observed at different attack stages and their frequency (how frequent) are important characteristics that help in identifying APT group's characteristics. Our similarity measure can identify these properties and evaluate the attribution similarity measures for a given attack pattern in terms of TTPs. We also present a structured and curated TTP-based attribution dataset, which is a stepping stone for cybersecurity researchers to perform research in APT attribution.
% CAPTAIN also enhances the state of the TTP extraction method by proposing TTPXHunter, which is based on a cyber-domain-specific SecureBERT model. We present a large augmented sentence-TTP dataset of $47,860$ sentences and their corresponding TTPs for TTP classification.
% 2. Performance\\
% TTPXHunter outperforms state-of-the-art TTP extraction models, and
The implemented experiments demonstrate how traditional similarity measures and available attribution literature struggle for the attribution task. The attribution performed using CAPTAIN outperforms all implemented attribution methods and state-of-the-art. Further, we look to identify if any more characteristics of TTPs found in the attack can help improve attribution and the efficacy of other attack patterns in performing attribution. 
APT attribution aids in the improvement of proactive defense against sophisticated cyber attacks, empowering the cybersecurity community to stay one step ahead. 
The CAPTAIN presents a foundational block for APT attribution using TTPs and their characteristics and can potentially provide a new direction for the research in APT attribution.

% \section{Introduction}
% \label{intro}
% Your text comes here. Separate text sections with
% \section{Section title}
% \label{sec:1}
% Text with citations \cite{RefB} and \cite{RefJ}.
% \subsection{Subsection title}
% \label{sec:2}
% as required. Don't forget to give each section
% and subsection a unique label (see Sect.~\ref{sec:1}).
% \paragraph{Paragraph headings} Use paragraph headings as needed.
% \begin{equation}
% a^2+b^2=c^2
% \end{equation}

% % For one-column wide figures use
% \begin{figure}
% % Use the relevant command to insert your figure file.
% % For example, with the graphicx package use
%   \includegraphics{example.eps}
% % figure caption is below the figure
% \caption{Please write your figure caption here}
% \label{fig:1}       % Give a unique label
% \end{figure}
% %
% % For two-column wide figures use
% \begin{figure*}
% % Use the relevant command to insert your figure file.
% % For example, with the graphicx package use
%   \includegraphics[width=0.75\textwidth]{example.eps}
% % figure caption is below the figure
% \caption{Please write your figure caption here}
% \label{fig:2}       % Give a unique label
% \end{figure*}
% %
% % For tables use
% \begin{table}
% % table caption is above the table
% \caption{Please write your table caption here}
% \label{tab:1}       % Give a unique label
% % For LaTeX tables use
% \begin{tabular}{lll}
% \hline\noalign{\smallskip}
% first & second & third  \\
% \noalign{\smallskip}\hline\noalign{\smallskip}
% number & number & number \\
% number & number & number \\
% \noalign{\smallskip}\hline
% \end{tabular}
% \end{table}

\section*{Conflicts of interest/Competing interests}
\textbf{Non-financial interests:} Author Sandeep Kumar Shukla is a member of the editorial board at the International Journal of Information Security. 
\\
\textbf{Financial interests:} This research is partially supported by C3i (Cyber Security and Cyber Security for Cyber-Physical Systems) Innovation Hub at IIT Kanpur. Nanda Rani received the Prime Minister Research Fellowship (PMRF), Ministry of Education, Government of India.

Apart from this, there is no other Competing Interest/Conflict of Interest.

\section*{Data availability}
The dataset utilized in this research is available at \url{https://github.com/nanda-rani/CAPTAIN}. Please note that this is just a preview of the dataset; the remaining part, along with the code, will be released following the acceptance of the research paper.

\begin{acknowledgements}
This work was partially supported by C3i (Cyber Security and Cyber Security for Cyber-Physical Systems) Innovation Hub at IIT Kanpur. Nanda Rani thanks the Prime Minister Research Fellowship (PMRF), Ministry of Education, Government of India, for the valuable financial support.
\end{acknowledgements}

\appendix
% \input{Appendix}

%% \appendix
\section{Appendix}
\label{sec:appendix}
\subsection{Properties of similarity measure} 
\label{sec:properties}
% (Lemma???)
The proposed similarity measure represents its effectiveness by exhibiting several properties as follows: 
\subsubsection{Non-negativity}
\label{subsubsec:NonNegativity}
This property represents that the similarity measure should always return a non-negative value for any pair of inputs, i.e., $Sim(seq_i,seq_j) \geq 0$.
In our similarity measure equation-\ref{eq:1}, $m$ and $n$ represents length of sequences i.e., $m,n>0$; therefore $m2^{m-1}+n2^{n-1}>0$. Also, every element of frequency and length vector are non-negative, i.e., $\lambda_i,\mu_i\geq 0\;\forall i$; therefore $\langle \mu,\lambda \rangle\geq 0$. When we combine both $m2^{m-1}+n2^{n-1}>0$ and $\langle \mu,\lambda \rangle\geq 0$, we conclude:

$$\boxed{Sim(seq_i,seq_j) \geq 0}$$ 

\subsubsection{Boundedness}
\label{subsubsec:Boundness}
The similarity measure must have a lower and upper bound, indicating the minimum and maximum possible similarity scores between TTP sequences.
    \begin{itemize}
        \item \textit{Lower Bound:} In the non-negative property, we demonstrate that $Sim(seq_a,seq_b) \geq 0$ means the lower bound can be a non-negative number. Suppose an extreme case when two sequences $seq_a$ and $seq_b$ have no common subsequence, then the vectors $\mu$ and $\lambda$ are zero vectors and $\langle \mu,\lambda \rangle = 0$. Combination of statements $Sim(seq_a,seq_b) \geq 0$ and $\langle \mu,\lambda \rangle = 0$ shows that the lowest value of $Sim(seq_a,seq_b)$ is zero. 
        \item \textit{Upper Bound:} 
        % In section~\ref{subsec:similarityproof}, We have proved that the maximum value of numerator of similarity score~\ref{eq:1} is maximum $m2^{m-1}+n2^{n-1}$  when both subsequences are the same, which is  the similarity measure is normalized within 
        For any two sequences $seq_a$ and $seq_b$ of length $m$ and $n$, where $m\leq n$; $\binom{m}{i}$ represents the number of $i-length$ subsequences from $m-length$ sequence. While $\mu_i$ is the number of common subsequences of length $\lambda_i$, common in both sequences, i.e., $\mu_i\leq \binom{m}{i}$. Therefore, the following inequality is satisfied.
        \begin{equation}
        \label{eq:5}
            \sum_{i=0}^{len(\lambda)}\lambda_i \mu_i \leq \sum_{i=0}^{m} i \binom{m}{i}
        \end{equation}

       LHS can be maximized to equal RHS when every possible subsequence is common with $seq_b$, i.e., if $seq_a\subseteq seq_b$. Similarly for sequence $seq_b$, $\sum_{i=0}^{len(\lambda)}\lambda_i \mu_i$ is maximum for $seq_b\subseteq seq_a$. Therefore, the maximum similarity score can be achieved when $seq_a = seq_b$. We can write from eq.~\ref{eq:5} for $seq_a = seq_b$ as follows:
        \begin{align*}
            \sum_{i=0}^{len(\lambda)}\lambda_i \mu_i &= \sum_{i=0}^{m} i \binom{m}{i} = m2^{m-1}
        \end{align*}
\text{Similarly,}
        \begin{align*}
            \sum_{i=0}^{n} i \binom{n}{i} = n2^{n-1}
        \end{align*}
        Using the above relation, we can write the upper bound (maximum value) of a similarity measure as follows:
        
        \begin{align*}
            Sim(seq_a,seq_b) &= \frac{2\langle \mu,\lambda \rangle} {m2^{m-1}+n2^{n-1}}\\
            &= \frac{2\sum_{i=0}^{len(\lambda)}\lambda_i \mu_i} {m2^{m-1}+n2^{n-1}}\intertext{When $seq_a = seq_b \text{ i.e.}, n=m,$}
            &= \frac{2n2^{n-1}} {n2^{n-1}+n2^{n-1}}\\
            &=1
        \end{align*}
       
       % $$\sum_{i=0}^{len(\lambda)}\lambda_i \mu_i = \sum_{i=0}^{m} i \binom{m}{i} = m2^{m-1}$$

        % $m2^{m-1}+n2^{n-1}$ is a constant 
    \end{itemize}
    Based on the above discussion, we can conclude that the proposed similarity measure is bounded between zero and one, i.e., $0\leq Sim(seq_a,seq_b)\leq 1$.

% \begin{figure*}
%     \centering
%     \includegraphics[width=\textwidth, height=3.5cm]{fig/TTP-APT_DATASET.pdf}
%     \caption{TTP-APT Dataset Preparation Flow}
%     \label{fig:ttpaptdataset}
% \end{figure*}

\subsubsection{Symmetry}
\label{subsubsec:symmetry}
According to this property, the similarity measure should return the same value when comparing sequence $seq_i$ to sequence $seq_j$ as it does when comparing sequence $seq_j$ to sequence $seq_i$, i.e., $Sim(seq_i,seq_j)=Sim(seq_j,seq_i)$. We use the commutative property of addition to prove this property.
    From eq.\ref{eq:1}, the similarity measure of the two sequences is as follows:
    \begin{align*} 
        Sim(seq_i,seq_j) &= \frac{2\langle \mu,\lambda \rangle} {m2^{m-1}+n2^{n-1}} \intertext{By commutative property of addition,}
        &= \frac{2\langle \mu, \lambda \rangle} {n2^{n-1}+m2^{m-1}}\\
        &=Sim(seq_j,seq_i)
    \end{align*}

\subsubsection{Sensitivity to Differences}
\label{subsubsec:senstivity}
The similarity measure should differentiate similar and dissimilar sequences effectively. It should assign higher similarity scores to more similar sequences and lower scores to less similar sequences. It is a common practice to show the similarities and differences property based on the possible cases. So, we discuss the following four cases to show that the proposed similarity measure is sensitive to difference.
    \begin{itemize}
        \item \textit{Case-1:} When a common TTP between two sequences changes to an uncommon TTP. In this case, at least one common subsequence of a certain length gets reduced, which reduces the corresponding frequency element. For example, if the changed element is part of a maximum one-length CSS, then the element of frequency vector corresponding to the one length is reduced by one. If the changed element is involved in a maximum two-length CSS, then the frequency vector corresponding to one and two lengths is reduced by one. If the changed element contributed in $n>2$ length CSS, it reduces every frequency vector one to $n$ length by at least one.
        % More frequencies can be reduced if the changed element is not an end element.
        Therefore, the similarity score reduces more when we change an important element, in this case.
        \item \textit{Case-2:} When a common TTP between two sequences changes to another common TTP. In this case, the changed element can break a larger CSS into smaller ones or merge two CSS to make them larger; breaking reduces various elements of frequency vector, and bridging increases. In this case, the similarity score decreases when the changed element breaks the link and increases when it makes the larger link as expected. 
        \item \textit{Case-3:} When an uncommon TTP between two sequences changes to an uncommon TTP. Changing an irrelevant element into another irrelevant element does not make any change.
        \item \textit{Case-4:} When an uncommon TTP between two sequences changes to a common TTP. In this case, the changed element increases the frequency of single-length CSS and can be added at the end of existing CSS, which also increases the frequency of an additional length; it can link two CSS to make a larger CSS that increases the frequency of an additional length and various pressing lengths. Therefore, the similarity measure increases when an irrelevant element changes to more relevant. 
    \end{itemize}

\subsection{Multi-tactics TTPs} 
\label{sec:TTPs}
\begin{table*}[!ht]
    \centering
    \caption{Multi-tactics TTPs}
    \begin{tabularx}{\linewidth}{cp{6cm}X}
    \hline
     \textbf{Technique ID} & \textbf{Technique Name} & \textbf{Tactics} \\
     \hline
     % \hline
     T1037 & Boot or Logon Initialization Scripts & Privilege Escalation, Persistence\\
     % \hline
     T1557 & Adversary-in-the-Middle & Collection, Credential Access\\
     % \hline
     T1543 & Create or Modify System Process & Privilege Escalation, Persistence\\
     % \hline
     T1133 & External Remote Services & Initial Access, Persistence\\
     % \hline
     T1547 & Boot or Logon Autostart Execution & Privilege Escalation, Persistence\\
     % \hline
     T1040 & Network Sniffing & Discovery, Credential Access \\
     % \hline
     T1053 & Scheduled Task/Job & Privilege Escalation, Execution, Persistence \\
     % \hline
     T1091 & Replication Through Removable Media & Lateral Movement, Initial Access \\
     % \hline
     T1659 & Content Injection & Initial Access, Command and Control \\
     % \hline
     T1055 & Process Injection & Privilege Escalation, Defense Evasion \\
     % \hline
     T1205 & Traffic Signaling & Command and Control, Defense Evasion, Persistence \\
     % \hline
     T1550 & Use Alternate Authentication Material & Lateral Movement, Defense Evasion \\
     % \hline
     T1610 & Deploy Container & Execution, Defense Evasion \\
     % \hline
     T1548 & Abuse Elevation Control Mechanism & Privilege Escalation, Defense Evasion \\
     % \hline
     T1542 & Pre-OS Boot & Defense Evasion, Persistence \\
     % \hline
     T1497 & Virtualization/Sandbox Evasion & Discovery, Defense Evasion \\
     % \hline
     T1072 & Software Deployment Tools & Lateral Movement, Execution \\
     % \hline
     T1098 & Account Manipulation & Privilege Escalation, Persistence \\
     % \hline
     T1574 & Hijack Execution Flow & Privilege Escalation, Defense Evasion, Persistence \\
     % \hline
     T1078 & Valid Accounts & Privilege Escalation, Initial Access, Defense Evasion, Persistence \\
     % \hline
     T1546 & Event Triggered Execution & Privilege Escalation, Persistence \\
     % \hline
     T1056 & Input Capture & Collection, Credential Access \\
     % \hline
     T1197 & BITS Jobs & Defense Evasion, Persistence \\
     % \hline
     T1134 & Access Token Manipulation & Privilege Escalation, Defense Evasion \\
     % \hline
     T1622 & Debugger Evasion & Discovery, Defense Evasion \\
     % \hline
     T1484 & Domain Policy Modification & Privilege Escalation, Defense Evasion \\
     \hline
\end{tabularx}
    \label{tab:multitacticsTTP}
\end{table*}
In the current version of the ATT\&CK framework, i.e., $v14$, there are $14$ tactics and $201$ techniques. There are a total of $27$ out of $201$ TTPs belonging to multiple tactics. The details of such TTPs are in Table~\ref{tab:multitacticsTTP}.

\end{document}